\documentclass[10pt]{article}

\pdfoutput=1

\RequirePackage{amsmath,amsthm,amssymb,cite,graphicx}

\RequirePackage{geometry}
\usepackage[colorlinks=true,allcolors={blue}]{hyperref}

\usepackage[margin=0pt]{subfig}

\numberwithin{equation}{section}

\title{Tube Formulae for Generalized von Koch Fractals through Scaling Functional Equations}
\author{Will Hoffer} 
\date{September 22, 2024}

\usepackage{fancyhdr}
\fancypagestyle{firstpage}{
    \fancyhf{}

    \fancyfoot[L]{
        \footnotesize\smash{
            \begin{tabular}[b]{@{}p{\linewidth}@{}}
                \bigskip 
        \end{tabular}}
    }
}

\newcommand{\CC}{\mathbb{C}}
\newcommand{\RR}{\mathbb{R}}
\newcommand{\ZZ}{\mathbb{Z}}
\newcommand{\NN}{\mathbb{N}}
\newcommand{\HH}{\mathbb{H}}
\newcommand{\Dd}{\mathcal{D}}
\newcommand{\Mm}{\mathcal{M}}
\newcommand{\Rr}{\mathcal{R}}
\newcommand{\e}{\varepsilon}

\newcommand{\nth}{^\text{th}}

\DeclareMathOperator{\Res}{Res}

\RequirePackage{dsfont}
\newcommand{\1}{\mathds{1}}

\RequirePackage{mathtools}
\DeclarePairedDelimiter{\set}{\{}{\}}
\DeclarePairedDelimiter{\norm}{\|}{\|}

\newcommand{\Mink}{\text{Mink}}

\newcommand{\tubezeta}{\widetilde\zeta}
\newcommand{\partialzeta}{\widetilde\xi} 
\newcommand{\dimension}{N}
\newcommand{\Cnr}{{C_{n,r}}}
\newcommand{\vkCurve}{{C'_{n,r}}}
\newcommand{\Knr}{{K_{n,r}}}
\newcommand{\vkSnow}{{K'_{n,r}}}

\input{content/thmSetup.tex}

\let\OLDthebibliography\thebibliography
\renewcommand\thebibliography[1]{
  \OLDthebibliography{#1}
  \setlength{\parskip}{0pt}
  \setlength{\itemsep}{0pt plus 0.3ex}
}

\begin{document}

\maketitle
\thispagestyle{firstpage}
\enlargethispage{-2\baselineskip} 

\begin{abstract}
%
%

In this work, we provide a treatment of scaling functional equations in a general setting involving fractals arising from sufficiently nice self-similar systems in order to analyze the tube functions, tube zeta functions, and complex dimensions of relative fractal drums. Namely, we express the volume of a tubular neighborhood in terms of scaled copies of itself and a remainder term and then solve this expression by means of the tube zeta functions. 

We then apply our methods to analyze these generalized von Koch fractals, which are a class of fractals that allow for different regular polygons and scaling ratios to be used in the construction of the standard von Koch curve and snowflake. In particular, we describe the volume of an inner tubular neighborhoods and the possible complex dimensions of such fractal snowflakes. 

\end{abstract}

%
%

\section{Introduction}
\label{sec:intro}
%
%

In this work, we establish tube formulae for a class of fractals called generalized von Koch snowflakes (see Definition~\ref{def:vkSnow}) such as those depicted in Figure~\ref{fig:threeGKFs}. We construct and analyze scaling functional equations in a general setting where a set is partitioned into finitely many scaled copies and a remainder term from the leftover region. Using the theory of tube zeta functions and complex fractal dimensions, we deduce the leading asymptotics of the tube function for such sets based on the order of this remainder term. 

Then, as an application of our main result, we deduce explicit tube formulae for the generalized von Koch fractals studied herein. In doing so, we extend the results of Lapidus and Pearse on computing a tube formula for the standard von Koch snowflake in \cite{LapPea06}. Additionally, our analysis of scaling functional equations with remainder terms can be applied to many other types of fractals whose tubular neighborhoods are not completely self-similar. 

Our work is organized as follows. We begin with a background regarding the origins of these fractals, previous work that we build upon, and then a summary of our main results in Section~\ref{sec:intro}. Next, in Section~\ref{sec:GKFs}, we introduce generalized von Koch fractals that are the source of our examples and the main application of results, and then we provide necessary background on tube functions, tube zeta functions, and complex dimensions relevant to our work in Section~\ref{sec:tubesAndZetas}. We state and prove our main results on the analysis of scaling functional equations (with error terms) first in a general setting (Section~\ref{sec:generalSFEs}) and then specializing to tube and zeta functions in Section~\ref{sec:fractalSFEs}. Finally, we apply our results to generalized von Koch fractals in Section~\ref{sec:appToGKFs}. 

\begin{figure}[t]
    \centering
    \subfloat{\includegraphics[trim=60 0 60 0,clip,width=3cm]{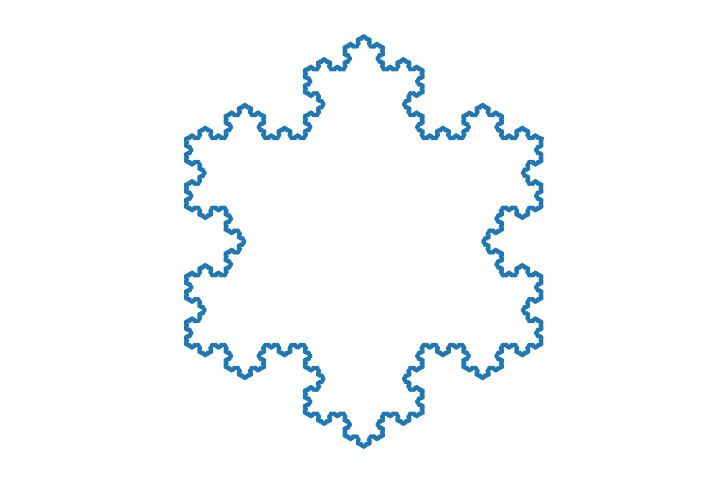}}
    \qquad 
    \subfloat{\includegraphics[trim=60 0 60 0,clip,width=3cm]{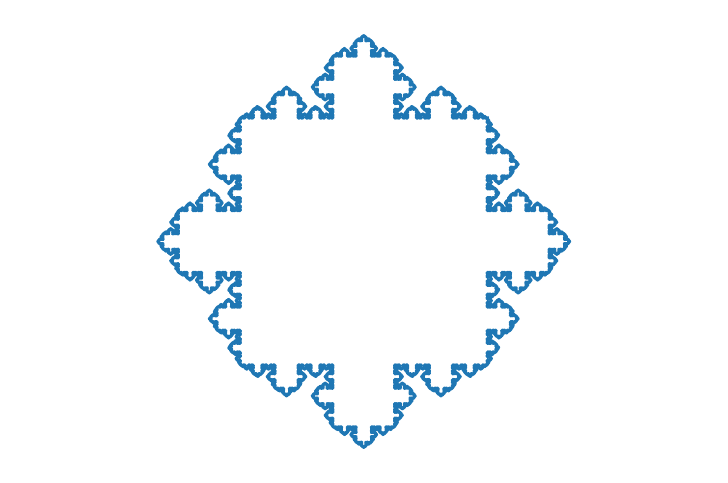}}
    \qquad 
    \subfloat{\includegraphics[trim=60 0 60 0,clip,width=3cm]{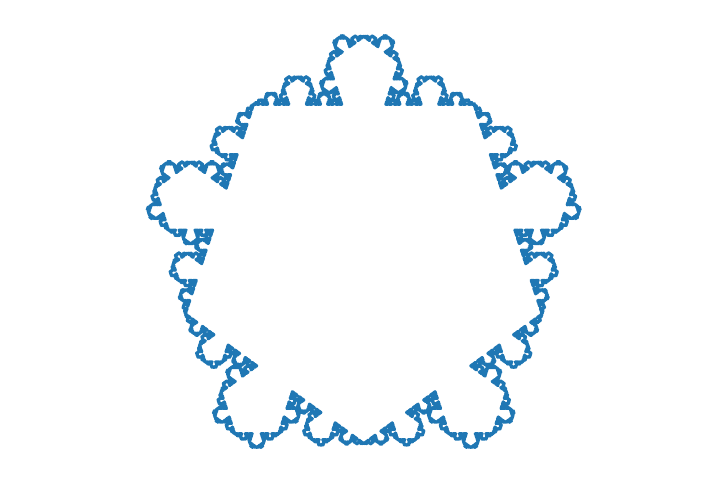}}
    \caption{The von Koch snowflake (left) and two of its generalizations, one with squares (middle) and the other with pentagons (right).}
    \label{fig:threeGKFs}
\end{figure}

\subsection{Background}
In 1904, Helge von Koch published his work on the construction of a planar curve without tangent lines at any point, describing the curve that now bears his name \cite{Koch1904, Koch1906}. See the left curve in Figure~\ref{fig:vKCurves} for a depiction. The union of three of these curves, placed about the edges of an equilateral triangle, from what is now called a von Koch snowflake\footnote{Interestingly, it may or may not have been von Koch himself who first made this combination of the curves leading to the snowflake shape. The earliest known reference to it appears as an exercise in a book published in 1912; see \cite{Dem23} for the reference and discussion of the history.} as seen in the leftmost shape in Figure~\ref{fig:threeGKFs}. The other two fractals depicted are what we call generalized von Koch fractals, and are the focus of this work.

A tube formula (see Definition~\ref{def:tubeFunction}) for the von Koch snowflake was established in the work of Lapidus and Pearse. They computed that the volume of an inner epsilon neighborhood of the von Koch snowflake takes the form:
\[ V(\e) = \sum_{n\in\ZZ}\phi_n \e^{2-D-in P} + \sum_{n\in\ZZ}\psi_n \e^{2-in P}, \]
where $D=\log_3 4$ is the Minkowski dimension of the snowflake, $P=2\pi/\log 3$ is the multiplicative period of the oscillations, and with $\psi_n,\phi_n$ as constants depending only on $n$ \cite{LapPea06}. As a consequence, they deduced the possible complex (fractal) dimensions of the von Koch snowflake. These complex dimensions encode both the amplitude and period of geometric oscillations in a fractal, and are pivotal to establishing such tube formulae in the work of Lapidus and collaborators; see \cite{LapvFr13_FGCD,LRZ17_FZF} and the references therein. 

Our method of studying functional equations is closely related to the methods developed in renewal theory. Feller established the renewal equation and methods thereabout in his work in queuing theory \cite{Fel1950}, and now renewal theory has many applications in fractal analysis such as in the work of Strichartz on self-similar measures \cite{Str1, Str2,Str3}, the work of Lapidus in \cite{LapDundee}, and the work of Kigami and Lapidus on the Weyl problem \cite{KiLap93}. Furthermore, the notion of a (scaling) functional equation was employed by Deniz, Ko\c cak, \"Ozdemir, and \"Ureyen in \cite{DKOU15} to provide a new proof of a tube formula for self-similar sprays in the work of Lapidus and Pearse \cite{LapPea10,LapPea11}. In regards to the same fractals studied herein, Michiel van den Berg and collaborators have used techniques from renewal theory to study the heat equation on these generalized von Koch fractals such as in \cite{vdBGil98,vdBHol99,vdB00}. 

In the study of fractals with multiple scaling ratios, there is a lattice/non-lattice dichotomy in behavior depending on whether or not the ratios are arithmetically related or not. This dichotomy has also been called the arithmetic/non-arithmetic dichotomy, and has been discussed in the work of Lalley in \cite{Lal1,Lal2} and his paper in \cite{BedKS}, in the work of Lapidus and collaborators (such as in \cite{LapvFr13_FGCD,LRZ17_FZF,LapDundee,LapPea10,LapPea11},) and for the generalized von Koch snowflakes by van den Berg and collaborators in their aforementioned work on heat content. See \cite{DGMRSurvey} and the references therein for more information about this dichotomy in fractal geometry.

\begin{figure}[t]
    \centering
    \subfloat{\includegraphics[width=4cm]{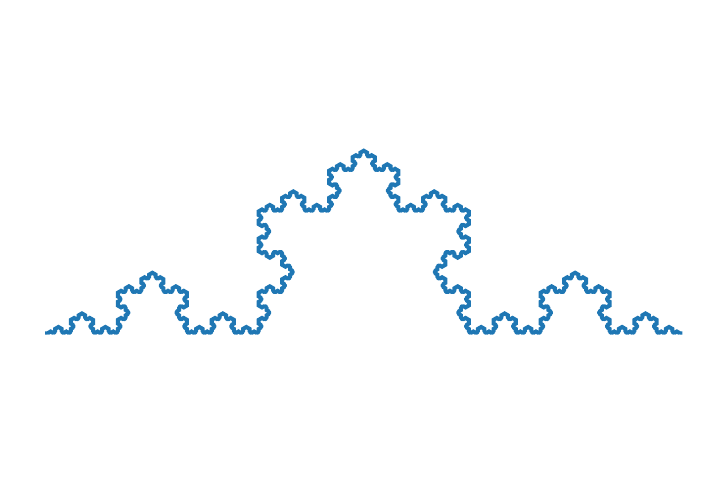}}
    \qquad
    \subfloat{\includegraphics[width=4cm]{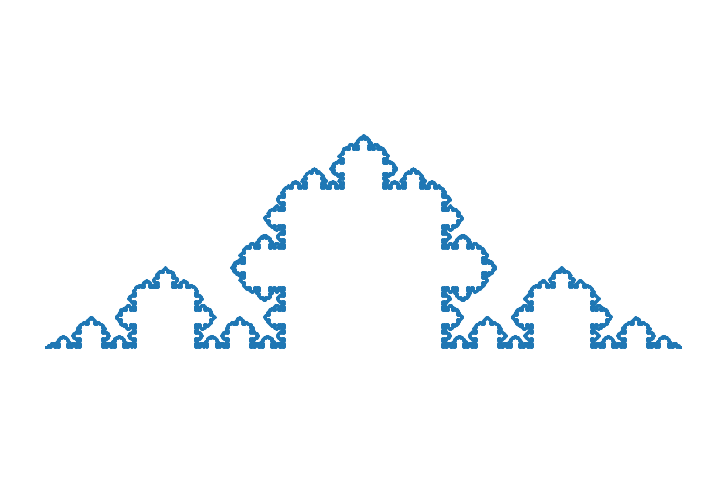}}
    \caption{Two fractal curves are depicted, that studied by von Koch (left) and a curve created from generalizing his construction (right).}
    \label{fig:vKCurves}
\end{figure}

\subsection{Main Results}
Our main results consist of establishing scaling functional equations and analyzing their solutions first in a general setting (Section~\ref{sec:generalSFEs}) and then specifically with relative tube and zeta functions (Section~\ref{sec:fractalSFEs}). We conclude by applying our work to generalized von Koch fractals (introduced in Section~\ref{sec:GKFs} and analyzed in Section~\ref{sec:appToGKFs}.) 

Firstly, Theorem~\ref{thm:transSFE} and Theorem~\ref{thm:solnSFE} represent our treatment of scaling functional equations of the form 
\[ f(x) = \sum_{i=1}^m a_i f(x/\lambda_i) +R(x), \]
where we include a remainder term $R$. Additionally, our treatment of the solution is through multiplicative means, using restricted Mellin transforms, rather than converting to an additive variable and using renewal theory directly. Letting $\Mm^\delta$ denote the restricted Mellin transform (c.f. Definition~\ref{def:truncatedMellin}), we show that 
\[ \Mm^\delta[f](s) = \frac1{1-\sum_{i=1}^m a_i \lambda_i^s}(E(s)+\Mm^\delta[R](s)), \]
where $E(s)=\sum_{i=1}^m a_i\lambda_i^s\Mm_\delta^{\delta/\lambda_i}[f](s)$ is an entire function and where $\Mm^\delta[R](s)$ is holomorphic in the right half plane $\HH_{\sigma_0}$ when $R(t)=O(t^{-\sigma_0})$. The function $f$ may be recovered by Mellin inversion and computation (or estimation) of the integrals appearing $E$; see Theorem~\ref{thm:solnSFE} and the discussion after its proof. 

Secondly, in the context of fractal geometry we describe relative fractal drums $(X,\Omega)$ that arise from self-similar iterated function systems $\Phi$ obeying the open set condition. We introduce a notion of the set $\Omega$ ``osculating'' $X$ under iteration by $\Phi$ so that points in $\phi(\Omega)$ remain closest to $\phi(X)$ for each $\phi\in\Phi$ (see the fifth condition in Definition~\ref{def:oscRFD}.) For any such relative fractal drums, we establish in Theorem~\ref{thm:SFEofSSS} a scaling functional equations satisfied by the tube function $V_{X,\Omega}=|X_\e\cap\Omega|$, namely 
\[ V_{X,\Omega}(\e) = \sum_{i=1}^m a_i\lambda_i^\dimension V_{X,\Omega}(\e/\lambda_i) + V_{X,R}(\e), \]
where $V_{X,R}(t)=O(t^{\dimension - \sigma_0})$ and where the set $\set{\lambda_i}$ is the set of distinct scaling ratios of maps in $\Phi$ and $a_i$ is the multiplicity of each scaling ratio. We then deduce Theorem~\ref{thm:SFEofSSS} that the tube zeta function $\tubezeta_{X,\Omega}$ (see Definition~\ref{def:tubeZeta}) takes the form 
\[ \tubezeta_{X,\Omega}(s) = \frac1{1-\sum_{i=1}^m a_i\lambda_i^s} \cdot h(s), \]
where $h(s)$ is holomorphic in the right half plane $\HH_{\sigma_0}$ (compare with Equation~\ref{eqn:zetaFE}.) Consequently, the complex dimensions of $X$ with real part strictly larger than $\sigma_0$ are exactly the solutions $\omega$ to the complexified Moran equation 
\[ 1 = \sum_{i=1}^m a_i \lambda_i^\omega \]
where also $h(\omega)\neq 0$; see Theorem~\ref{thm:cdimsFromSFE} and Corollary~\ref{cor:exactCdims}.

We conclude by returning to the main class of examples in this work, generalized von Koch snowflakes. These are described in Section~\ref{sec:GKFs} and see Figure~\ref{fig:threeGKFs} for some examples. One way to construct such fractals is to begin with a regular $n$-gon. Then, for each edge between vertices, replace the middle $r\nth$ portion with the $n-1$ edges of a smaller regular $n$-gon to construct the first prefractal approximation. Repeating this process yields the fractal $\Knr$, though to be precise we define these fractals through iterated function systems. We deduce in Theorem~\ref{thm:cDimsOfGKFs} that the possible complex dimensions of $\Knr$ (with positive real part) are solutions $\omega$ to the equation 
\[ 1 = 2\Big( \frac{1-r}2 \Big)^\omega + (n-1)r^\omega. \]
Furthermore, we show that the $k\nth$ antiderivative $V_{\Knr,\Omega}^{[k]}$ of the tube zeta functions of $\Knr$ relative to the interior region $\Omega$ is defines takes the form 
\[ V_{\Knr,\Omega}^{[k]}(t) = \sum_{\mathclap{\omega\in\Dd_{\Knr,\Omega}(\HH_0)}} \Res\Bigg( \frac{t^{2-s+k}}{(3-s)_k} \tubezeta_{\Knr,\Omega}(s;\delta);\omega \Bigg) + O(t^{2-\alpha+k}),  \]
where we refer the reader to Theorem~\ref{thm:GKFtubeFormula} and Section~\ref{sec:appToGKFs} for the appropriate definitions and constraints. If the poles of the function $\tubezeta_{\Knr,\Omega}(s;\delta)$ are simple, then this formula takes the form:
\[ V_{\Knr,\Omega}^{[k]}(t) = \sum_{\mathclap{\omega\in\Dd_{\Knr,\Omega}(\HH_0)}}  \Res(\tubezeta_{\Knr,\Omega}(s;\delta);\omega)\frac{\omega^{2-s+k}}{(3-\omega)_k} + O(t^{2-\alpha+k}). \]
In both of these formulae, $(s)_k$ is the Pochhammer symbol defined by 
\[ (s)_k := \Gamma(s+k)/\Gamma(s) = s(s+1)(s+2)\cdots (s+k-1) \] 
and $\alpha>0$ is any (small) positive alteration to the exponent. 

Of note, in the case of generalized von Koch fractals we see the appearance of the lattice-non lattice dichotomy: the geometry of the fractal is fundamentally different depending on the arithmetic properties of the scaling ratios. In studying this problem through the lens of fractal zeta functions, we are able to understand this behavior by way of the structure of the complex dimensions of the fractal. See for example Figure~\ref{fig:cDimPlots2D}, where the standard von Koch snowflake is arithmetic while the ``squareflake" and ``pentaflake" fractals depicted are non-arithmetic. For the former, the tube function will have an oscillatory leading term, but for the other two the tube formula will have a monotone leading order term and oscillatory effects at lower orders.

\section{Generalized von Koch Fractals}
\label{sec:GKFs}
%
%

What we shall dub generalized von Koch fractals (abbr. GKFs) are fractals in which regular polygons other than triangles are used in the construction of the snowflake and/or those in which the scaling ratio chosen is a factor other than one third. This provides a two-parameter family of fractal curves, where $n\geq 3$ is the number of sides of the regular polygon and $0<r<1$ is the scaling ratio for the middle segments. 

\subsection{Generalized von Koch Fractal Curves}
To define an $(n,r)$-von Koch curve, we shall employ the notion of an iterated function system, which were introduced by Hutchinson in \cite{Hut81}. Hutchinson in particular proved that iterated function systems induce a contraction on the space of nonempty, compact subsets of $\RR^\dimension$ equipped with the Hausdorff metric, whence by the Picard-Banach fixed point theorem there is a unique fixed point of the system (which is nonempty and compact.) This fixed point, or equivalently the attractor of the system, is a convenient way to define many fractals \cite{Fal90_FG,Bar88}. 

We shall actually use a slightly more restrictive type of iterated function system, namely a self-similar system, whose definition we provide below.
\begin{definition}[Self-Similar System]
    A self-similar system $\Phi$ on a complete metric space $(X,d)$ is a finite set of contractive similitudes:
    \[ \Phi := \set{\phi_k:X\to X}_{k=1}^m, \]
    where for each $k=1,...,m$ and for every $x,y\in X$,
    \[ d(\phi_k(x),\phi_k(y)) = r_k d(x,y), \]
    where $r_k\in(0,1)$ is the scaling ratio of $\phi_k$.    
\end{definition}
Note that an equivalent phrasing of the above statement is that a self-similar system is an iterated function system where each contraction mapping is, a fortiori, a similitude. 

Similitudes in Euclidean spaces are compositions of translations, rotations, reflections, and homotheties. In order to explicitly write self-similar systems in what follows, we introduce the following translation, rotation, and homothety/scaling transformations of $\RR^2$: 
\begin{equation}
    \label{eqn:transforms}
    \begin{aligned}
        T_{(a,b)}(x,y)  &:= (x+a,y+b),                                              &(a,b)  &\in\RR^2;  \\
        R_\theta(x,y)   &:= (x\cos\theta-y\sin\theta, x\sin\theta+y\cos\theta),     &\theta &\in\RR;    \\
        S_\lambda(x,y)  &:= (\lambda x, \lambda y),                                 &\lambda&\in\RR^+. 
    \end{aligned}
\end{equation}
Note that the scaling ratio of $S_\lambda$ is $\lambda$, and that translations and rotations are isometries, viz. their scaling ratios are equal to one.

Lastly, we shall also need some angles related to a regular $n$-gon. Let us define $\theta_n=\frac{2\pi}n$ to be the exterior angle (or equivalently the central angle) and $\alpha_n=\pi-\frac{2\pi}n$ to be the interior angle of a regular $n$-gon. See Figure~\ref{fig:ngonAngles} for a depiction of these angles on a hexagon.
\begin{figure}[t]
    \centering
    \includegraphics[width=0.4\linewidth]{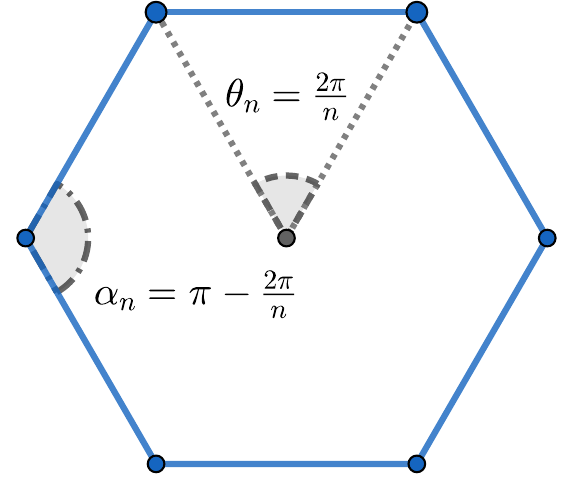}
    \caption{A depiction of the central angle $\theta_n=2\pi/n$ and the interior angle $\alpha_n=\pi-2\pi/n$ of a regular $n$-gon, illustrated on a hexagon where $n=6$.}
    \label{fig:ngonAngles}
\end{figure}

With all of this geometric information, we may now explicitly write a self-similar system whose attractor will be an $(n,r)$-von Koch curve. We shall also allow for sets which are isometric to such an attractor to be considered generalized von Koch curves as well. Let the mappings be defined as follows:
\begin{equation}
    \label{eqn:vkCurveMaps}
    \begin{aligned}
        \phi_L      &:= S_\ell, \\
        \phi_R      &:= T_{(\ell+r,0)}\circ S_\ell, \\
        \psi_1      &:= T_{(\ell,0)}\circ R_{\alpha_n} \circ S_r, \\
        \psi_k      &:= T_{\psi_{k-1}(1,0)}\circ R_{\alpha_n-(k-1)\theta_n}\circ S_r,\quad k>1. 
    \end{aligned}
\end{equation}
Note that $T_{(a,b)},S_\lambda,$ and $R_\theta$ are transformations of $\RR^2$ as defined in Equation~\ref{eqn:transforms} and $\theta_n$ and $\alpha_n$ are, respectively, the central and interior angles of a regular $n$-gon such as those depicted in Figure~\ref{fig:ngonAngles} for $n=6$.
\begin{definition}[$(n,r)$-von Koch Curve]
    \label{def:vkCurve}
    Let $n\geq 3$ be an integer, let $r\in\RR$ satisfy $0<r<1$, and let $\ell=\frac{1-r}2$.  

    A set $\vkCurve\subset\RR^2$ is said to be an $(n,r)$\textbf{-von Koch curve} if it is isometric to the set $\Cnr\subset\RR^2$ which is the unique, nonempty, compact fixed point associated to the following self-similar system on $\RR^2$:
    \begin{align*}
        \Phi_{n,r}  &:= \set{ \phi_L, \phi_R, \psi_k:\RR^2\to\RR^2, k=1,...,n-1 }, 
    \end{align*}   
    where the mappings $\phi_L,$ $\phi_R,$ and $\psi_k$ (for $k=1$ through $n-1$) are defined as in Equation~\ref{eqn:vkCurveMaps}.
\end{definition}
In other words, we may write that 
    \[ \vkCurve \cong \Cnr = \bigcup_{\phi\in\Phi_{n,r}}\phi[\Cnr]. \]
If we denote by $F:\RR^2\to\RR^2$ the isometry that sends $\Cnr$ to $\vkCurve$, then we may write a self-similar system for $\vkCurve$ using:
\[ F\circ \Phi_{n,r} := \set{F\circ \phi : \phi\in \Phi_{n,r}}. \]
Note that there can be more than one iterated function system that generates the same set $\Cnr$ or $\vkCurve$.

An algorithmic approach to constructing the curve is given by iterating the system $\Phi_{n,r}$ on the unit interval $[0,1]\times\set{0}$. The first step removes the the middle $r\nth$ piece of the interval on the $x$-axis and adjoins the $n-1$ other sides of a regular $n$-gon with length $r$. Each successive step repeats this process on every line segment, again removing the middle $r\nth$ portion of the line and adjoining the edges of a polygon whose side length is $r$ times that of the line segment's length. The regular polygon is always added with the same orientation with respect to the line segment. 

Given $n$ total $(n,r)$-von Koch curves, a generalized von Koch snowflake is simply the union of these curves placed about the edges of a regular $n$-gon (of side length one). For example, the leftmost fractal in Figure~\ref{fig:threeGKFs} is a union of three copies of the left curve in Figure~\ref{fig:vKCurves} and the middle fractal in Figure~\ref{fig:threeGKFs} is a union of four copies of the right curve in Figure~\ref{fig:vKCurves}. So, $(3,\frac13)$-von Koch snowflake is the ``ordinary'' von Koch snowflake depicted leftmost in Figure~\ref{fig:threeGKFs}. Additionally, this figure shows a ``squareflake'' (the $(4,\frac14)$-von Koch snowflake) and a ``pentaflake'' (the $(5,\frac15)$-von Koch snowflake), which are generalized snowflake fractals with fourfold and fivefold symmetry, respectively. The general definition is as follows.
\begin{definition}[$(n,r)$-von Koch Snowflake]
    \label{def:vkSnow}
    Let $n\geq 3$ be an integer, let $r\in \RR$ with $0<r<1$, and let $\Cnr$ be an $(n,r)$-von Koch curve with endpoints $(0,0)$ and $(1,0)$.

    A set $\vkSnow\subset\RR^2$ is an $(n,r)$\textbf{-von Koch snowflake} if it is isometric to the set
    \begin{align*}
        \Knr = \bigcup_{k=1}^n U_k[\Cnr],
    \end{align*} 
    where $U_1:= R_{\theta_n}\circ T_{(1,0)}$ and $U_k := R_{k \theta_n}\circ T_{U_{k-1}(1,0)}$ for $k>1$. Note that $T_{(a,b)}$ and $R_\theta$ are transformations of $\RR^2$ defined in Equation~\ref{eqn:transforms}.
\end{definition}
Equivalently, one could construct an explicit self-similar system for $\Knr$ using the maps of $\Cnr$ composed with the isometries $U_k$, $k=1,...,n$ and defined $\Knr$ to be the unique attractor of this system. As later in the paper we shall use symmetry to focus on the portion of this fractal closest to a single curve, this distinction will not be necessary.

\begin{figure}
    \centering
    \subfloat{\includegraphics[trim= 60 0 60 0, clip, width=3cm]{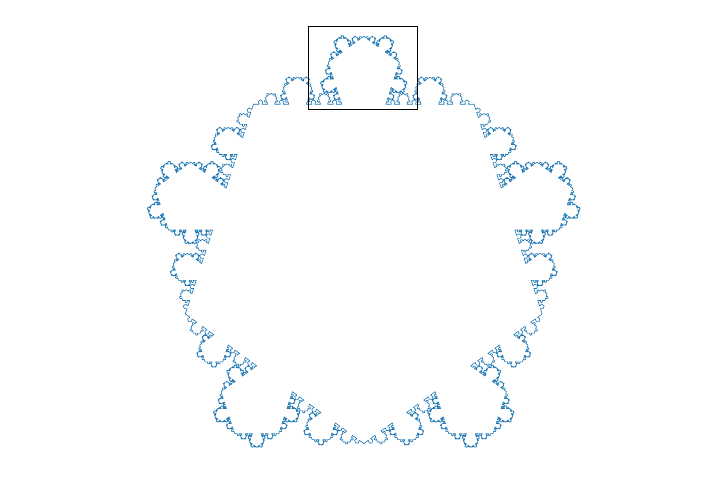}}
    \qquad
    \subfloat{\includegraphics[width=3cm]{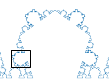}}
    \qquad
    \subfloat{\includegraphics[width=3cm]{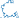}}
    \caption{A depiction of $K_{5,\frac15}$ (at the fourth stage of the prefractal approximation) together with two zoomed-in images of the pentagonal frills.}
    \label{fig:pentaflakeZoom}
\end{figure}

For a closer look at these fractal snowflakes, note that Figure~\ref{fig:pentaflakeZoom} depicts a prefractal approximation of the pentaflake $K_{5,\frac15}$ with two extra levels of zoom onto one of the fringes, which can be seen to be pentagons.

Additionally, it is of note that we are identifying these von Koch snowflakes as unions of curves. One might instead define a snowflake to be the region(s) enclosed by the union of $(n,r)$-von Koch curves, in which case what we call the snowflakes here would be the boundary of this set. In this setting, it would be suitable to call $\Knr$ an $(n,r)$-von Koch \textit{snowflake boundary} or a \textit{snowflake curve} to be unambiguous. In this work, the distinction shall not be necessary.

The author first encountered these in the work of M. van den Berg and his collaborators on asymptotics of heat content of GKFs such as in \cite{vdBHol99}. These have been studied by other authors such as Paquette and Keleti \cite{KelPaq10}, who in particular describe when such GKFs are topologically simple curves or not. They established the condition stated in Proposition~\ref{prop:selfAvoid} for self-avoidance of the curve. Figure~\ref{fig:intersectingGKFs} depicts three $(6,r)$-von Koch snowflakes with decreasing values of $r$; the boundary of the curve may intersect if $r$ is large but cannot when $r$ is sufficiently small.
\begin{proposition}[Self-avoidance of GKFs \cite{KelPaq10}]
    \label{prop:selfAvoid}
    An $(n,r)$-von Koch curve is non-self-intersecting if the scaling ratio $r>0$ satisfies the following:
    \begin{align*} 
        r &< \frac{\sin^2(\pi/n)}{\cos^2(\pi/n)+1}, &&\text{ if }n\text{ is even, and }\\
        r &< 1-\cos(\pi/n),                         &&\text{ if }n\text{ is odd.}
    \end{align*}
    The boundary of the corresponding $(n,r)$-von Koch snowflake $\Knr$ is topologically simple under these conditions.
\end{proposition}
In particular, this proposition implies that for each $n\geq 3$, there is an interval of admissable scaling ratios, namely $(0,r_0)$. The condition in Proposition~\ref{prop:selfAvoid} is sufficient but not necessary, as the fractals can interleave for certain values of $n$ and $r$. We refer the interested reader to Keleti and Paquette's paper \cite{KelPaq10} for more information. Of note, $n=6$ is the first value for which $1/n$ exceeds the value of $r_0$ provided in the proposition. See Figure~\ref{fig:intersectingGKFs} for a depiction of $K_{6,\frac16}$.

\begin{figure}[t]
    \centering
    \subfloat{\includegraphics[trim=80 0 80 0,clip,width=3cm]{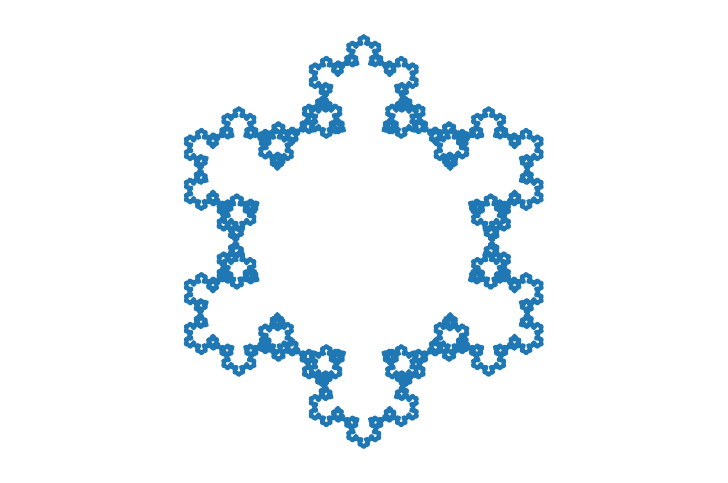}}
    \qquad
    \subfloat{\includegraphics[trim=80 0 80 0,clip,width=3cm]{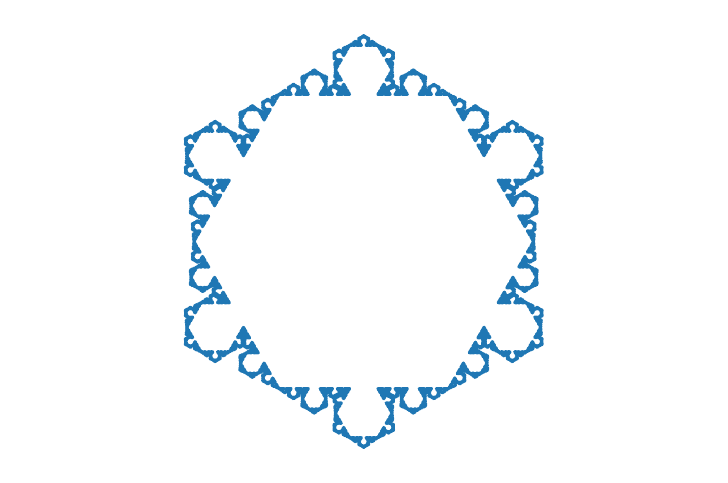}}
    \qquad
    \subfloat{\includegraphics[trim=80 0 80 0,clip,width=3cm]{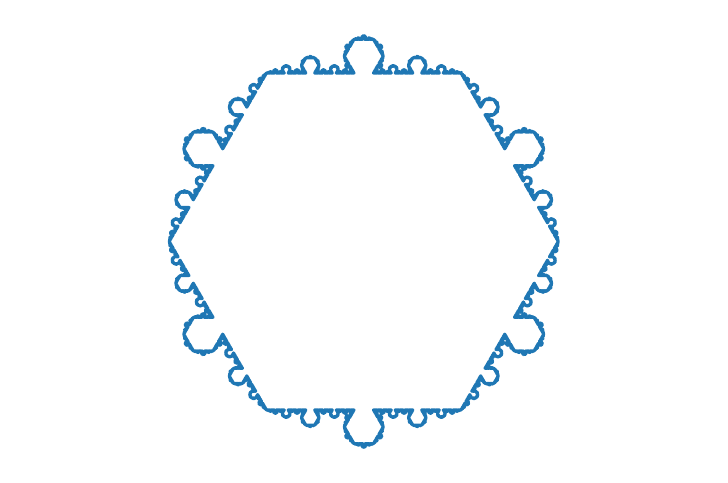}}
    \caption{Three $(6,r)$-von Koch snowflakes, with values of $r$ equal to $0.3$, $\frac16$, and $0.1$ respectively from left to right. The fractal curve is topologically simple if $r<1-\frac{\sqrt{3}}2$, such as with the rightmost figure. For the other two fractals, depicted in the middle and on the left, the curves are self-intersecting.}
    \label{fig:intersectingGKFs}
\end{figure}

\section{Tubular Neighborhoods and Zeta Functions}
\label{sec:tubesAndZetas}
%
%

Tube zeta functions were introduced in \cite{LRZ17_FZF} to study fractals in higher dimensions and to extend the theory of complex dimensions (see Definition~\ref{def:cDims}) to such fractals. These tube zeta functions are constructed from tube functions, which are the volumes of tubular neighborhoods of the given set. The tube zeta functions are essentially restricted Mellin transforms of this volume function, and thus, they possess important scaling properties (c.f. Lemma~\ref{lem:zetaScaling}.)  

\subsection{Tube Functions and their Properties}
Let $\dimension$ denote a given Euclidean dimension. Given a set $X\subset\RR^\dimension$, let $X_\e$ denote an epsilon neighborhood of $X$, that is, 
\[ X_\e := \{ y\in\RR^\dimension : \exists x\in X,\,|y-x|<\e \}. \]
Note that even if $X$ itself is not Lebesgue measurable, $X_\e$ is necessarily measurable as it is open. For example, it may be written as a union of open sets $B_\e(x)$ for $x\in X$, where $B_\e(x)$ is an open ball of radius $\e$ centered at $x$.

In what follows, we will be considering these tubular neighborhoods relative to some other set $\Omega\subset\RR^\dimension$. To that end, we recall the notion of a relative fractal drum \cite{LRZ17_FZF}. 
\begin{definition}[Relative Fractal Drum]
    \label{def:RFD}
    Let $X,\Omega\subset\RR^\dimension$ and suppose further that $\Omega$ is open, has finite Lebesgue measure, and has the property that $\exists\delta>0$ such that $\Omega_\delta\supset X$. Then the pair $(X,\Omega)$ is called a relative fractal drum, or RFD for short. 
\end{definition}

The tube function of an RFD $(X,\Omega)$ shall be the volume of the tubular neighborhood of $X$ contained within the set $\Omega$. Notably, such tube functions may be defined for any subset of $\RR^\dimension$ so long as the set $\Omega$ relative to which the volume is computed is Lebesgue measurable.
\begin{definition}[Relative Tube Function of a Set]
    \label{def:tubeFunction}
    Let $X\subset\RR^\dimension$, let $\Omega\subset\RR^\dimension$ be a Lebesgue measurable set, and denote by $m$ the $\dimension$-dimensional Lebesgue measure. The tube function $V_{X,\Omega}$ of $X$ relative to $\Omega$ is the function $V_{X,\Omega}(\e):=m(X_\e\cap \Omega)$, defined for $\e\geq 0$.
\end{definition}
When the set $\Omega$ is clear from context or is all of $\RR^\dimension$, we may suppress it in the notation. Any such function $V_{X,\Omega}$ is continuous and non-decreasing on its domain, and it is finite when $X$ or $\Omega$ is a bounded set.

Note that the Lebesgue measure $m$ has the scaling property that, for any $\lambda\in\RR^+$, $m(\lambda X) = \lambda^\dimension m(X)$. When a tubular neighborhood $X_\e$ is scaled, the parameter $\e$ defining the new scaled neighborhood will scale linearly. In other words, $\lambda\cdot(X_\e)=(\lambda X)_{\lambda \e}$. Combining these properties, we may deduce how tube functions change under scaling. In the simplest case, where $\Omega=\RR^\dimension$ and we omit it, we have that $V_{\lambda X}(\lambda\e) = \lambda^\dimension V_{X}(\e)$. We state and prove the following property for relative tube functions. 
\begin{lemma}[Scaling Property of Relative Tube Functions] 
    \label{lem:volumeScaling}
    Let $X\subset\RR^\dimension$ and let $\Omega\subset \RR^\dimension$ be measurable. For any $\e\geq 0$ and $\lambda>0$, we have that
    \[ V_{\lambda X,\lambda\Omega}(\e) = \lambda^\dimension V_{X,\Omega}(\e/\lambda). \]
    More generally, if $\phi:\RR^\dimension\to\RR^\dimension$ is a similitude with scaling ratio $\lambda$, then we have the analogous identity:
    \[ V_{\phi(X),\phi(\Omega)}(\e) = \lambda^\dimension V_{X,\Omega}(\e/\lambda). \]
\end{lemma}
\begin{proof}
    Let $\phi=\widetilde\psi\circ\phi_\lambda\circ\psi'$, where $\psi',\widetilde\psi$ are compositions of rotations, reflections, and translations and where $\phi_\lambda$ is a scaling transformation, viz. $\phi_\lambda(x)=\lambda x$. Any similitude of scaling ratio $\lambda$, by definition, may be written this way. Note that the first identity is a special case of the second property; merely set $\widetilde\psi=\psi'=\text{Id}$. For convenience, let us denote $U'=\psi'(U)$ in what follows, and note that $\lambda U = \phi_\lambda(U)$.
    
    The maps in question have lots of nice properties. First, we note that $\psi'$ and $\widetilde\psi$ are isometries, which is to say that $m(U')=m(U)$ and $m(\tilde\psi(U))=m(U)$ for any $U\subset\RR^\dimension$. Since $\widetilde\psi,\psi',\phi_\lambda,$ and their composition $\phi$ are injective, they each have the property that the intersection of the images of two sets is the image of the intersection of those sets.
    The same property is true (unconditionally) for the union of images being the image of the union. 

    This lattermost fact is one way to see that $\widetilde\psi(X_\e)=(\widetilde\psi(X))_\e$. One may write the neighborhood $X_\e$ as a union of $\e$-balls, and then apply the preservation of unions under images to see that this set is exactly the neighborhood of the transformed set $\widetilde\psi(X)$. 
    
    Using this fact and the above properties, we find that the tube functions are unaffected by an injective isometry such as $\widetilde\psi$:
    \begin{align*}
        V_{\widetilde\psi(\lambda X'),\widetilde\psi(\lambda \Omega')}(\e) 
            &= m(\widetilde\psi((\lambda X')_\e)\cap\widetilde\psi(\lambda\Omega')) \\
            &= m(\widetilde\psi[(\lambda X')_\e\cap\lambda\Omega'])
                =V_{\lambda X',\lambda\Omega'}(\e).
    \end{align*}
    As mentioned before, we have that $\phi_\lambda(X_\e)=(\lambda X)_{\lambda\e}$. Writing $\e=\lambda t$, we obtain that:
    \begin{align*}
        V_{\lambda X',\lambda\Omega'}(\lambda t) = m(\phi_\lambda(X'_t\cap\Omega')) =\lambda^N V_{X',\Omega'}(t),
    \end{align*}
    where the last step is the scaling property of the Lebesgue measure. We have already shown that for an injective isometry like $\psi'$, $V_{\psi'(X),\psi'(\Omega)}=V_{X,\Omega}$. Thus, the result follows by writing $t=\e/\lambda$ and combining these equalities. 
\end{proof}

\subsection{Tube Zeta Functions and Complex Dimensions}
Tube zeta functions were introduced in \cite{LRZ17_FZF} to study fractals in higher dimensions. They are defined as a restricted Mellin transform of the tube function of the corresponding set, introducing a cutoff value of $\delta>0$ to the upper bound of the usual Mellin transform.
\begin{definition}[Relative Tube Zeta Function]
    \label{def:tubeZeta}
    Let $(X,\Omega)$ be a relative fractal drum and let $\delta>0$. The relative tube zeta function $\tubezeta_{X,\Omega}$ of $X$ relative to $\Omega$ is given by
    \[ \tubezeta_{X,\Omega}(s;\delta) := \int_0^\delta t^{s-N-1}V_{X,\Omega}(t)\,dt, \]
    for $s\in\CC$ with sufficiently large real part. 
\end{definition}
As with the tube functions, when the set $\Omega$ is clear from context, we shall sometimes write $\tubezeta_X=\tubezeta_{X,\Omega}$ and omit the relative set. We shall also denote by $\tubezeta_{X,\Omega}$ the maximal meromorphic extension of the holomorphic function defined by the integral. 

Next, we observe that for any two $\delta_2>\delta_1>0$, we have that 
\begin{equation}
    \label{eqn:zetaDiff}
    \tubezeta_{X,\Omega}(s;\delta_2)-\tubezeta_{X,\Omega}(s;\delta_1)
        = \int_{\delta_1}^{\delta_2}t^{s-N-1}V_{X,\Omega}(t)\,dt. 
\end{equation}
Of note, this difference is a holomorphic function which possesses an entire analytic continuation \cite{LRZ17_FZF}. As such, the poles of a tube zeta function are not dependent on the choice of $\delta$. 

These poles of a tube zeta function are of importance to the geometry of the fractal. They are called complex (fractal) dimensions of the set $X$. Typically, one specifies a subset of the complex plane, or a ``window,'' which is usually the domain of meromorphicity of the function in $\CC$ or a half-plane contained therein.
\begin{definition}[Complex Dimensions of a Set]
    \label{def:cDims}
    Let $(X,\Omega)$ be a relative fractal drum and let $\tubezeta_{X,\Omega}$ be the relative tube zeta function of $X$.
    
    If $W\subset \CC$, then the complex dimensions of $X$ relative to $\Omega$ contained in the window $W$, denoted by $\Dd_X(W)$, are the poles of $\tubezeta_{X,\Omega}$ contained within $W$. 
\end{definition}
The complex dimensions control the exponents and the form of the tube formula for the set \cite{LapvFr13_FGCD,LRZ17_FZF}. See especially chapter five in \cite{LRZ17_FZF}.

Tube zeta functions are well suited to studying self-similar or nearly self-similar objects due to their scaling properties. These follow from the nature of the zeta function as a (restricted) Mellin transform on the space of positive scaling factors. The scaling relation is complicated only slightly by the truncation of the integral, as the cutoff value will change with the scaling. 
\begin{lemma}[Scaling Property of $\tubezeta_{X,\Omega}$]
    \label{lem:zetaScaling}
    Let $(X,\Omega)$ be a relative fractal drum in $\RR^\dimension$ and let $\lambda>0$. Then for all $s$ in its domain, the tube zeta function of $X$ relative to $\Omega$ satisfies the following:
    \[ \tubezeta_{\lambda X,\lambda\Omega}( s;\delta) = \lambda^{s} \tubezeta_{X,\Omega}(s;\delta/\lambda). \]
    Moreover, if $\phi:\RR^\dimension\to\RR^\dimension$ is a similitude of scaling ratio $\lambda>0$, then similarly we have that:
    \[ \tubezeta_{\phi(X),\phi(\Omega)}(s;\delta) = \lambda^s \tubezeta_{X,\Omega}(s;\delta/\lambda). \] 
\end{lemma}
We note that this result is essentially just Proposition~4.6.11 in \cite{LRZ17_FZF}, where the only change is that we draw explicit attention to the usage of similitudes. We include a proof since this result may be seen as a corollary of Lemma~\ref{lem:volumeScaling}. 
\begin{proof}
    Using Lemma~\ref{lem:volumeScaling} and a change of variables, we compute that
    \begin{align*}
        \tubezeta_{\phi(X),\phi(\Omega)}( s;\delta) 
            &= \int_0^{\delta} t^{s-N}\lambda^\dimension V_{X,\Omega}(t/\lambda)\,\frac{dt}{t} \\
            &= \int_0^{\delta/\lambda} (\lambda t)^{s-N}\lambda^\dimension V_{X,\Omega}(t)\,\frac{dt}{t} 
            = \lambda^{s} \tubezeta_{X,\Omega}(s;\delta/\lambda).
    \end{align*}
\end{proof}
Now, from Lemma~\ref{lem:zetaScaling} and Equation \ref{eqn:zetaDiff}, we obtain the functional relation:
\begin{equation}
    \label{eqn:zetaFunctionalEq1}
    \tubezeta_{\lambda X,\lambda\Omega}(s;\delta) = \lambda^s \tubezeta_{X,\Omega}(s;\delta) 
        + \lambda^s\int_{\delta}^{\delta/\lambda}t^{s-N-1}V_{X,\Omega}(t)\,dt.
\end{equation}
We shall introduce notation for this ``partial'' tube zeta function, as it encapsulates the effects the truncation on the scaling property of the tube zeta function.
\begin{definition}[Partial Tube Zeta Function]
    \label{def:partialTubeZeta}
    Let $(X,\Omega)$ be a relative fractal drum in $\RR^\dimension$, $V_{X,\Omega}$ its relative tube function, and $0< \delta_1<\delta_2$. We define a partial tube zeta function to be the following:
    \[ 
        \partialzeta_{X,\Omega}(s;\delta_1,\delta_2) 
            := \int_{\delta_1}^{\delta_2} t^{s-\dimension-1}V_{X,\Omega}(t)\,dt, 
    \] 
    for all $s\in\CC$ with sufficiently large real part.
\end{definition}
These partial tube zeta functions shall appear in scaling functional equations such as Equation~\ref{eqn:zetaFunctionalEq1}, and thus, we shall have need of estimating them. We note the following. Letting $\sigma=\Re(s)$, we have that:
\begin{align}
    \label{eqn:partialTubeEstimate}
    |\partialzeta_{X,\Omega}(s;\delta_1,\delta_2)|\leq 
        \begin{cases}
            V_{X,\Omega}(\delta_2) \cfrac{\delta_2^{\sigma-\dimension}-\delta_1^{\sigma-\dimension}}{\sigma-\dimension}, & \sigma\neq \dimension, \\
            V_{X,\Omega}(\delta_2) \log(\delta_2/\delta_1), & \sigma=\dimension.
        \end{cases}
\end{align}
Note that these bounds are enough to deduce that $\partialzeta_{X,\Omega}$ extends to an entire function in the complex variable $s\in\CC$, whence follows the result of Lapidus, Radunovi\'c, and \v Zubrini\'c on the independence of the poles of $\tubezeta_{X,\Omega}$ on the parameter $\delta$ in \cite{LRZ17_FZF}. We conclude by stating this as a lemma:
\begin{lemma}[Partial Tube Functions are Entire]
    \label{lem:partialTubeEntire}
    Let $\partialzeta_{X,\Omega}(s;\delta_1,\delta_2)$ be a partial tube zeta function of the relative fractal drum $(X,\Omega)$, with fixed positive parameters $\delta_1,\delta_2$. Then $\partialzeta_{X,\Omega}$ extends to an entire function in the variable $s$.
\end{lemma}

\section{General Scaling Functional Equations}
\label{sec:generalSFEs}
%
%

In this section, we study scaling functional equations in a general setting. In particular, we show how they may be solved by means of (truncated) Mellin transforms. The approach is similar in spirit to solving functional equations by means of renewal theory, except in this case we use the Mellin transform which is natural for studying for the positive real line as a group with multiplication. We shall illustrate the ideas up through the use of the transform to create nicely solvable functional equations into the Mellin inversion theorem. Notably, we use a truncation of the Mellin transform as opposed to the standard Mellin transform. This is primarily because in the application to tube functions in Section~\ref{sec:fractalSFEs}, these are preferred. 

The positive real line $\RR^+$ (as a group with multiplication) may be thought of as the space of scaling factors, and its associated Haar measure $\frac{dt}{t}$ is invariant with respect to scaling transformations. We shall first treat scaling functional equations of functions on this space, and later specialize to tube and zeta functions in the next section. 

\subsection{Truncated Mellin Transforms}
Let $C^0(\RR^+)$ denote the space of continuous functions $f:\RR^+\to\RR$. The standard Mellin transform of $f$ is the integral $\Mm[f](s) = \int_0^\infty x^{s-1}f(x)\,dx.$ A truncated or restricted Mellin transform is simply an integral of the same integrand but over an interval which is a subset of $(0,\infty)$. 
\begin{definition}[Truncated Mellin Transform]
    \label{def:truncatedMellin}
    Let $f\in C^0(\RR^+,\RR)$ and fix $\alpha,\beta \geq 0$. The truncated Mellin transform of $f$, denoted by $\Mm_\alpha^\beta[f]$, is given by 
    \[ \Mm_\alpha^\beta[f](s) := \int_\alpha^\beta t^{s}f(t)\,\frac{dt}t \]
    for all $s\in\CC$ for which the integral is convergent. If $\alpha=0$, we write $\Mm^\beta$ for the truncated transform and if additionally $\beta=\infty$ then $\Mm=\Mm_0^\infty$ is the standard Mellin transform.
\end{definition}
Note that we may equivalently define $\Mm_\alpha^\beta[f]$ to be the Mellin transform of $f$ times the characteristic function of the interval $(\alpha,\beta)$, viz. 
\[ \Mm_\alpha^\beta[f] = \Mm[f\cdot\1_{[\alpha,\beta]}]. \] 
This allows us to compare the convergence of the two directly, and it shows that the truncated transform inherits the properties of its counterpart, most notably linearity. 

First in our discussion of convergence of such transforms, we note an immediate corollary of the alternate definition is that a restricted transform converges automatically if the Mellin transform converges. It is known (see for example Chapter 6 in \cite{Graf10}) that $\Mm[f](s)$ is holomorphic in the vertical strip $\sigma_-<\Re(s)<\sigma_+$, where 
\begin{align*}
    \sigma_- :=& \inf\set{\sigma : f(x) = O(x^{-\sigma}, x\to0^+)}\\
    \sigma_+ :=& \sup\set{\sigma : f(x) = O(x^{-\sigma}, x\to\infty)}.
\end{align*}
When $\alpha<\beta<\infty$, $f\cdot\1_{[\alpha,\beta]}\equiv 0$ as $x\to\infty$, whence $\sigma_+=\infty$. Similarly, if $0<\alpha<\beta$ then $f\cdot\1_{[\alpha,\beta]}\equiv 0$ as $x\to0^+$, whence $\sigma_-=-\infty$. So, if $0<\alpha<\beta<\infty$, the restricted transform is automatically entire, i.e. holomorphic in $\CC$. If $0=\alpha<\beta<\infty$ and $f=O(x^{\sigma_0})$, then $\Mm^\beta[f](s)$ is holomorphic in the half plane $\HH_{-\sigma_0}$, i.e. when $\Re(s)>-\sigma_0$.  

Notably, the restricted transform may converge even if the full Mellin transform integral diverges. Perhaps the most important class of functions for which this occurs is that of polynomials: if $f(t)=t^k$ and $\Re(s)>-k$, $\Mm^\beta[f](s)$ converges whilst $\Mm[f](s)$ is divergent for all $s\in\CC$. Thus, for a general polynomial $p(t)=\sum_{k=0}^n a_kt^k$, we have that 
\[ \Mm^\beta[p](s) = \sum_{k=0}^n \frac{a_k}{s+k}\beta^{s+k},\quad  \Re(s)>\max\set{-k: a_k\neq 0}. \]
Negative powers become admissable if the lower bound $\alpha$ of the truncation is positive. 

\subsection{Scaling Operators and Associated Zeta Functions}
Next, we define scaling operators which act on $C^0(\RR^+)$. To borrow terminology from statistical mechanics, we shall define pure and mixed scaling operators. A pure scaling operator $M_\lambda$, where $\lambda\in\RR^+$, shall act by precomposition by a scaling operator in the following fashion
\[ M_\lambda[f](x) := f(x/\lambda), \]
where this convention of inverting the scaling factor shall be convenient for our applications. A mixed scaling operator shall be a linear combination of such scaling operators. If we have that $L=\sum_{i=1}^m a_i M_{\lambda_i}$ is a (mixed) scaling operator, then it acts by:
\[ L[f](x) = \sum_{i=1}^m a_i f(x/\lambda_i). \]
We shall not require the combination to be convex; later, the only constraint we will add is that the multiplicities be positive and integral. For now, each $a_i\in\RR$. 

Given such a scaling operator $L$, we define an associated scaling zeta function. This function shall play a key role in describing solutions to scaling functional equations, owing to its relation to the Mellin transform of such functions and the role of its poles. 
\begin{definition}[Associated Scaling Zeta Function]
    \label{def:scalingOpZF}
    Let $L=\sum_{i=1}^m a_i M_{\lambda_i}$ be a scaling operator with $\lambda_i\in\RR^+$ and $a_i\in\RR$ for each $i=1,...,m$. 

    We define the zeta function $\zeta_L$ associated to $L$ to be
    \[ \zeta_L(s) := \cfrac1{1-\sum_{i=1}^m a_i \lambda_i^s}, \]
    for all $s\in \CC$ for which the expression converges. 
\end{definition}
In practice, we shall enforce that $0<\lambda_i<1$ and that $a_i>0$ for each $i$. Under these conditions, we obtain the characterization of such zeta functions in Proposition~\ref{prop:propertiesSZF}. Of note, the results are the exact same as those for self-similar fractal strings, c.f. Chapters~2 and 3 and in particular Theorem~3.6 in \cite{LapvFr13_FGCD}. See Figure~\ref{fig:cDimPlots} for two examples of such zeta functions which we will relate to von Koch fractals later in Section~\ref{sec:appToGKFs}.

To state these results, we shall define the similarity dimension of a scaling operator. 
\begin{definition}[Similarity Dimension of a Scaling Operator]
    \label{def:simDimScalingOp}
    Let $L$ be a scaling operator with scaling ratios $\lambda_i\in(0,1)$, each having positive multiplicities $a_i$. The similarity dimension $D\in\RR$ of $L$ is the unique real solution to the Moran scaling equation, that is $D$ solves the equation
    \begin{equation}
        \label{eqn:Moran}
        1 = \sum_{i=1}^m a_i \lambda_i^s. 
    \end{equation}
\end{definition}
If $\sum_{i=1}^m a_i >1$ and $m>2$, then $D$ is positive; this is automatically true if the multiplicities are integral and there are either at least two distinct scaling ratios or at least one scaling ratio with multiplicity larger than one. 

To see that Equation~\ref{eqn:Moran} has a unique real solution, define the function $p(s)=\sum_{i=1}^m a_i\lambda_i^s$. Then $p(0)=\sum_{i=1}^m a_i$, $\lim_{s\to-\infty}p(s)=+\infty$, and $\lim_{s\to\infty}p(s)=0$ since $0<\lambda_i<1$. Additionally, a simple calculation shows that $p'(s)<0$ for all $s\in\RR$ since $\log(\lambda_i)<0$ and $a_i>0$ for each $i$. Thus, it follows that $p(s)=1$ has exactly one solution for $s\in\RR$, and this solution is positive if $p(0)=\sum_{i=1}^m a_i >1$. This is guaranteed the multiplicities $a_i$ are integral and supposing that $L$ has at least two distinct scaling ratios or one scaling ratio with multiplicity of at least two.

\begin{figure}[t]
    \centering
    \subfloat{\includegraphics[width=4cm]{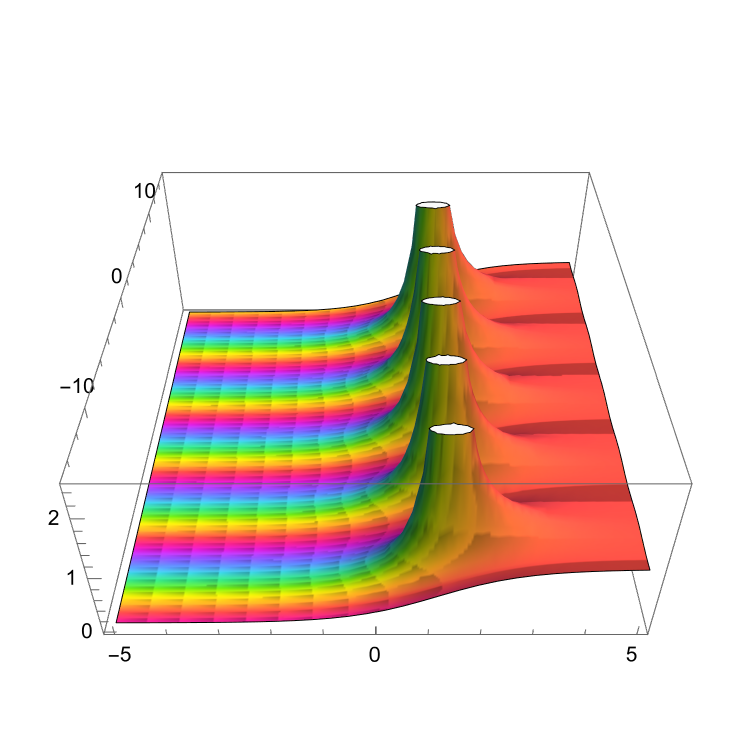}}
    \qquad
    \subfloat{\includegraphics[width=4cm]{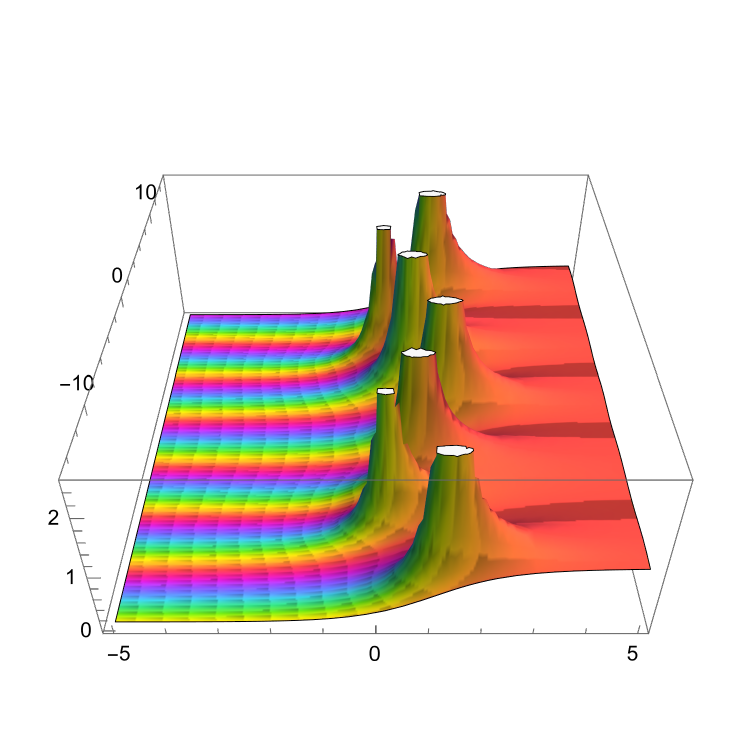}}
    \caption{Plots of associated scaling zeta functions of the scaling operator $L_{3,\frac13}:=4 M_{\frac13}$ (left) and the operator $L_{4,\frac14}:= 2M_{\frac{3}8}+3M_{\frac14}$ (right.) We will show in Section~\ref{sec:appToGKFs} that the possible complex dimensions of the fractals $K_{3,\frac13}$ and $K_{4,\frac14}$, respectively, occur at the poles of these functions.}
    \label{fig:cDimPlots}
\end{figure}

\begin{proposition}[Properties of Scaling Zeta Functions]
    \label{prop:propertiesSZF}
    Let $\zeta_L$ be the zeta function associated to the scaling operator $L$ with scaling ratios $\lambda_i\in(0,1)$ and positive multiplicities $a_i$. Let $D$ be the similarity dimension of $L$, defined as per Definition~\ref{def:simDimScalingOp}.
    
    Then $\zeta_L$ is holomorphic in the right half plane $\HH_D$, having a unique real pole at $D$. It possesses a meromorphic continuation to all of $\CC$, with poles at points the set $\Dd_L=\Dd_L(\CC)$ defined by:
    \[ \Dd_L(W) := \Big\{\omega\in W\subset\CC: 1-\sum_{i=1}^m a_i \lambda_i^s = 0 \Big\}. \]
    $\Dd_L(W)$ is the set of complex dimensions associated to $L$ in the window $W$.
\end{proposition}
\begin{proof}
    Clearly, $\zeta_L$ has poles at points $s\in\CC$ for which Equation~\ref{eqn:Moran} holds. That $D$ is the unique real pole of $\zeta_L$ follows from the preceding discussion that the Moran equation has a unique real solution. 

    Letting $p(s)=\sum_{i=1}^m a_i\lambda_i^s$ as before, we see that the zeta function is holomorphic when $\sigma=\Re(s)>D$ since 
    \[ \left|\sum_{i=1}^m a_i\lambda_i^s \right| \leq \sum_{i=1}^m a_i \lambda_i^\sigma = p(\sigma) < p(D) = 1. \]
    Thus $z=p(s)$ lies within the unit disk wherein $\frac1{1-z}$ is holomorphic. That $\zeta_L$ has a meromorphic extension to $\CC$ follows by analytic continuation.
\end{proof}

Lastly, we shall have need of growth estimates for this zeta function $\zeta_L$. There are two versions, the first of which is languid growth; see Definition~5.1.3 in \cite{LRZ17_FZF} and Definition~5.2 in \cite{LapvFr13_FGCD}. The second condition, which we shall use in this paper, is known as strongly languid growth, which implies languid growth; see Definition~5.1.4 in \cite{LRZ17_FZF} and Definition~5.3 in \cite{LapvFr13_FGCD}. We provide the definition for strongly languid growth for an function in general; a relative fractal drum $(X,\Omega)$ or a scaling operator $L$ is said to the strongly languid if the corresponding zeta function ($\tubezeta_{X,\Omega}$ or $\zeta_L$, respectively) is strongly languid (for some $\delta>0$ in the former case.)

There are two hypothesis, the first of which concerns power-law/polynomial growth along a sequence of horizontal lines.
\begin{definition}[Languidity Hypothesis \textbf{L1}]
    \label{def:languidL1}
    Let $f$ be a complex function which is holomorphic in the half plane $\HH_D$ and which possesses analytic continuation to a neighborhood of the set $\set{s\in\CC : \Re(s)>S(\Im(s))}$, where $S$ is a Lipschitz continuous function called a ``screen.'' Then $f$ is said to satisfy the languidity hypothesis \textbf{L1} with exponent $\kappa$ (with respect to the screen $S$) if the following holds: 

    There is a positive constant $C>0$, a constant $\beta>D$, and a doubly infinite sequence $\set{\tau_n}_{n\in\ZZ}$ with $\lim_{n\to\infty}\tau_n\to\infty$, $\lim_{n\to-\infty}\tau_n=-\infty$, and $\forall n\geq 1$, $\tau_{-n}<0<\tau_n$ so that on every horizontal interval of the form 
    \[ I_n = [S(\tau_n) +i\tau_n,\,\beta +i\tau_n ]\subset\CC, \]
    $f$ has at most power-law $\kappa$ growth with respect to $\tau_n$. More precisely, for all $\sigma$ in $[S(\tau_n),\beta]$, we have that 
    \[ |f(\sigma+i\tau_n)| \leq C(|\tau_n|+1)^\kappa. \]
    
\end{definition}
The second hypothesis concerns power-law/polynomial growth along a vertical curve. The standard version of languidity concerns growth along a single curve, and strong languidity concerns growth along a sequence of such curves moving to the left. These curves are known as screens owing to their connection to the windows appearing in Definition~\ref{def:cDims}. 
\begin{definition}[Languidity Hypothesis \textbf{L2}]
    \label{def:languidL2}
    Let $f$ be a complex function which is holomorphic in the half plane $\HH_D$ and which possesses analytic continuation to a neighborhood of the set $\set{s\in\CC : \Re(s)>S(\Im(s))}$, where $S$ is a Lipschitz continuous function called a ``screen.'' Then $f$ is said to satisfy the languidity hypothesis \textbf{L2} with exponent $\kappa$ (with respect to the screen $S$) if the following holds:  

    There is a positive constant $C$ such that for all $\tau \in\RR$ with $|\tau|\geq 1$, 
    \[ |f(S(\tau)+i\tau)| \leq C|\tau|^\kappa.   \]

\end{definition}
A function which satisfies both hypotheses \textbf{L1} and \textbf{L2} with respect to the same screen $S$ and exponent $\kappa$ is said to have languid growth (with respect to these parameters.) 

Strong languidity may be seen as a sequence of languidity conditions. Namely, it concerns languidity along a sequence of screens converging to the left. 
\begin{definition}[Strong Languidity of a Function]
    \label{def:stronglyLanguid}
    Let $f$ be a complex function which is holomorphic in the half plane $\HH_D$ and which possesses analytic continuation to the whole complex plane excepting for a discrete set of singular points. Then $f$ is said to be strongly languid with exponent $\kappa$ if the following conditions hold. 

    In what follows, we suppose the existence of a sequence $\set{S_m}_{m\in\NN}$ of sequence of Lipschitz continuous functions called ``screens" with two properties: firstly, letting $\norm{S_m}_\infty$ denote the supremum of $S_m$ over $\RR$, we have that $\sup_{m\in\NN} \norm{S_m}_\infty=-\infty$; secondly, there exists a uniform bound for the Lipschitz constants of the screens $S_m$.
    \begin{itemize}
        \item[\textbf{L1}] $f$ satisfies languidity hypothesis \textbf{L1} with respect to $\kappa$ and each screen $S_m$. (Equivalently, with respect to the replacement $S(\tau)\equiv -\infty$.)
        \item[\textbf{L2}$'$] There exist positive constants $C$ and $B$ such that for all $\tau\in\RR$ and $m\geq 1$, $f$ satisfies 
        \[ |f(S_m(\tau)+i\tau)|\leq CB^{|S_m(\tau)|}(|\tau|+1)^\kappa. \]
    \end{itemize}
    
\end{definition}
Note that the strong languidity condition allows for a prefactor with exponential growth related to $\norm{S_m}_\infty$, but is otherwise analogous to \textbf{L2}. Condition \textbf{L2}$'$ implies \textbf{L2} for each of the screens with finite supremum. Also, without loss of generality, the constant $C$ may be chosen to be the same in each of the hypotheses. 

Additionally, we note that typically it is the geometric object (e.g. a relative fractal drum, a fractal string, or similar) which is said to be languid. In the case of associated scaling zeta functions, the corresponding object would be the scaling operator itself. However, we are extending this notion to apply to a function which satisfies the corresponding growth conditions. We note that this will be useful later in reference to a function whose corresponding geometric object may or may not be known. See in particular Corollary~\ref{thm:languidZetaGKF}.

Now, we have that any of the associated scaling zeta functions $\zeta_L$ shall be strongly languid with respect to the exponent $\kappa=0$.
\begin{proposition}[Languidity of Scaling Zeta Functions]
    \label{prop:languidSZF}
    Let $L:=\sum_{i=1}^m a_i M_{\lambda_i}$ be scaling operator with distinct scaling ratios $\lambda_i\in(0,1)$ and positive multiplicities $a_i$ and suppose that $\zeta_L$ is its associated scaling zeta function. Let $\Lambda=\min_{i=1,...,m}\lambda_i$ and let $D$ be the similarity dimension of $L$. 
    
    Then $\zeta_L$ is strongly languid with respect to exponent $\kappa=0$ as in Definition~\ref{def:stronglyLanguid}. Explicitly, there is a constant $C>0$ and a sequence $\set{\tau_n}_{n\in\ZZ}$ with  $\lim_{n\to\infty}\tau_n\to\infty$, $\lim_{n\to-\infty}\tau_n=-\infty$, and $\tau_{-n}<0<\tau_n$ for all $n\geq 1$ such that the following conditions hold.
    \begin{enumerate}
        \item For any fixed $\alpha,\beta$ with $\alpha<D<\beta$, we have that $\zeta_L$ is uniformly bounded all of the intervals $I_n=[\alpha+i\tau_n,\beta+i\tau_n]$, with constant $C$. That is, for all $\tau_n$ and $\sigma\in[\alpha,\beta]$, $ |\zeta_L(\sigma+i\tau_n)|\leq C.$ 
        \item For all $s\in\CC$ with sufficiently small $\sigma=\Re(s)$, $|\zeta_L(s)|\leq C \Lambda^{|\sigma|}$.
    \end{enumerate} 
\end{proposition}
\begin{proof}
    This result is a corollary of Theorem~3.26 and the discussion in Section~6.4 of \cite{LapvFr13_FGCD}, as $\zeta_L$ is essentially the same as that of a self-similar fractal string. (Let the length and gap parameters be zero.) Note in particular the explicit estimate in Equation~6.36 of \cite{LapvFr13_FGCD}.
\end{proof}

\subsection{Scaling Functional Equations}
Let $f,R\in C^0(\RR^+)$. The function $f$ is said to obey a scaling functional equation on $\RR^+$ with remainder term $R$ if we have that for all $x\in\RR^+$, there is a scaling operator $L$ such that 
\begin{align}
    \label{eqn:SFE}
    f(x) = L[f](x) + R(x) = \sum_{i=1}^m a_i f(x/\lambda_i) + R(x).
\end{align}
The scaling functional equation is exact if $R\equiv 0$ and approximate if not. We generally assume that $R$ cannot be written as the image of $f$ under some scaling operator $T$ (whence the equation is exact with respect to the operator $L+T$,) but do not require it since it may be unknown except perhaps for estimates thereupon. 

We now analyze the solutions to scaling functional equations using (truncated) Mellin transforms. The first step is to transform the functional equation into an equation for the truncated Mellin transform.
\begin{theorem}[Truncated Mellin Transforms of SFEs]
    \label{thm:transSFE}
    Let $f,R\in C^0(\RR^+)$ and let $L=\sum_{i=1}^m a_i M_{\lambda_i}$ be a scaling operator such that the scaling ratios $\lambda_i$ are each in $(0,1)$ and with positive integer multiplicities $a_i$. Suppose also that either $m\geq 2$ or that $a_1\geq 2$. Additionally, assume that $R(x)=O(x^{-\sigma_0})$ as $x\to 0^+$ for some $\sigma_0\in \RR$.
    
    If $f$ satisfies the scaling functional equation $f=L[f]+R$, then its truncated Mellin transform $\Mm^\delta[f]$ satisfies the equation:
    \[ \Mm^\delta[f](s) = \zeta_L(s) (E(s)+\Mm^\delta[R](s)), \]
    for any $s\in\CC\setminus\Dd_L\cap \HH_{\sigma_0}$. Here, $E$ is the entire function explicitly given by:
    \[ E(s) := \sum_{i=1}^m a_i \lambda_i^s \Mm_{\delta}^{\delta/\lambda_i}[f](s), \]
    and $\HH_{\sigma_0}$ is the half plane $\Re(s)>\sigma_0$. 

    Moreover, if $D$ is the unique positive solution to Equation~\ref{eqn:Moran}, then the function $\Mm^\delta[f]$ is holomorphic in $\HH_{\max(D,\sigma_0)}$. It is meromorphic in $\HH_{\sigma_0}$, with poles contained within a subset of $\Dd_L(\HH_{\sigma_0})$.
\end{theorem}
\begin{proof}
    We begin by applying $\Mm^\delta$ to the scaling functional equation \[f=L[f]+R.\] This and a change of variables yield
    \begin{align*}
        \Mm^\delta[f](s) &= \sum_{i=1}^m a_i \int_0^{\delta} t^s f(t/\lambda_i)\,\frac{dt}t + \Mm^\delta[R](s) \\
            &= \sum_{i=1}^m a_i \lambda_i^s \int_0^{\delta/\lambda_i} t^s f(t)\,\frac{dt}t + \Mm^\delta[R](s).
    \end{align*}
    By the linearity of integration, we may split the integrals to write that 
    \begin{align*}
        \Mm^\delta[f](s) - \sum_{i=1}^m a_i \lambda_i^s \Mm^\delta[f](s) = E(s) + \Mm^\delta[R](s).
    \end{align*}
    Provided that $s\notin\Dd_L$, it follows that $\Mm^\delta[f](s)=\zeta_L(s)(E(s)+\Mm^\delta[R](s))$. 

    We have that the function $E$ is entire since it is a linear combination of entire functions $\lambda_i^s \Mm_\delta^{\delta/\lambda_i}[f]$ and that $\Mm^\delta[R](s)$ is holomorphic in $\HH_{\sigma_0}$. By Proposition~\ref{prop:propertiesSZF}, $\zeta_L$ is holomorphic in the half plane $\HH_D$ and meromorphic in all of $\CC$. 
    
    So, the functions $\zeta_L\cdot\Mm^\delta[R]$ and $\zeta_L\cdot E$ are holomorphic when $\Re(s)>\max{D,\sigma_0}$ and meromorphic in $\HH_{\sigma_0}$ with poles at exactly the points $\omega\in\CC$ such that 
    \[ \Res(\zeta_L(E(s)+\Mm^\delta[R]))\neq 0. \] 
\end{proof}

The solution to this scaling functional equation is then produced by the Mellin inversion theorem. This result is analogous to the Fourier inversion theorem, and provides a means by which one may invert the Mellin transform; see for instance \cite{Tit1937}. We shall make use of the following notation convention: 
\[ \int_{c-i\infty}^{c+i\infty} f(z)\,dz := \lim_{T\to\infty}\int_{c-iT}^{c+iT}f(z)\,dz, \]
which may also be considered a principal value integral such as in \cite{Graf10}.
\begin{theorem}[Solutions to Scaling Functional Equations]
    \label{thm:solnSFE}
    Let $f,R\in C^0(\RR^+)$, let $L=\sum_{i=1}^m a_i M_{\lambda_i}$ be a scaling operator, and suppose that $f$ satisfies the scaling functional equation $f=L[f]+R$. Suppose further that $R(x)=O(x^{-\sigma_0})$ as $x\to0^+$ and that $D$ is the unique positive solution to Equation~\ref{eqn:Moran}. 
    
    Then for any $x\in\RR^+$ and for any $\delta>x$, $f$ is given by:
    \begin{align*}
        f(x)    &= \Mm^{-1}[\zeta_L\cdot(E+\Mm^\delta[R])](x)\\
                &= \frac{1}{2\pi i} \int_{c-i\infty}^{c+i\infty} x^{-s} \zeta_L(s)(E(s)+\Mm^\delta[R](s))\,ds,
    \end{align*}
    where $c$ is any positive constant greater than both $D$ and $\sigma_0$ and $E$ is as in Theorem~\ref{thm:transSFE}.
\end{theorem} 
\begin{proof}
    Theorem~\ref{thm:transSFE} yields that $\Mm^\delta[f](s)=\zeta_L(s) (E(s)+\Mm^\delta[R](s))$. We know from this theorem and from Proposition~\ref{prop:propertiesSZF} that the functions $\zeta_L$ and $\Mm^\delta[R]$ are holomorphic in the half planes $\HH_D$ and $\HH_{\sigma_0}$, respectively. 
    
    Since $\Mm^\delta[f]=\Mm[f\cdot\1_{[0,\delta]}]$ is holomorphic in the half plane $\HH_{\max(D,\sigma_0)}$, we may apply the Mellin inversion theorem to find that for any $c>\max(D,\sigma_0)$,
    \[ 
        f(x)\cdot\1_{[0,\delta]}(x) 
            = \frac1{2\pi i}\int_{c-i\infty}^{c+i\infty} x^{-s} \zeta_L(s)(E(s)+\Mm^\delta[R](s))\,ds. 
    \]
    Since $\delta>x$, we have that $\1_{[0,\delta]}(x)=1,$ whence the result follows.
\end{proof}

In practice, the next step is to apply the residue theorem and contour deformation in order to obtain the function in terms of the poles of $\zeta_L$ and integration involving the remainder term. As this has already been done in detail for fractal tube functions (see Chapters 5 and 8 in \cite{LapvFr13_FGCD} in the context of explicit formulae and Chapter 5 of \cite{LRZ17_FZF} for the treatment of higher dimensional tube formulae,) we shall specialize to the relative tube functions of interest in the next section.

\section{Fractal Scaling Functional Equations}
\label{sec:fractalSFEs}
%
%

In this section, we establish our main results for fractal geometry: analysis of scaling functional equations for tube functions and explicit description of the corresponding tube zeta functions. We shall focus on relative fractal drums which describe attractors of self-similar systems, i.e. fractals arising from iterated function systems in which the mappings are all similitudes. The relevant background information may be found in Section~\ref{sec:tubesAndZetas}, and we shall use the theorems of Section~\ref{sec:generalSFEs} which have been established in a more general setting. 

\subsection{Osculant Fractal Drums}
Throughout this section, let $\Phi=\set{\phi_i}_{i=1}^m$ be a self-similar system on $\RR^\dimension$. For convenience, given a set $U\subset\RR^\dimension$ let us denote $U_i=\phi_i(U)$ for each $i=1,...,m$. Then if $X$ is the attractor of $\Phi$, we have that $X=\cup_{i=1}^m \phi_i(X)$ and in this case $X$ is said to be self-similar.

We will study iterated function systems satisfying the open set condition, introduced by Moran \cite{Moran46} for relating similarity and Hausdorff dimensions. See also the work of Hutchinson in \cite{Hut81}. 
\begin{definition}[Open Set Condition]
    \label{def:OSC}
    An iterated function system $\Phi=\set{\phi_i}_{i=1}^m$ on $\RR^\dimension$ satisfies the open set condition (OSC) if there exists an open set $U\subset\RR^\dimension$ such that:
    \begin{enumerate}
        \item $U\supset \bigcup\limits_{i=1}^m \phi_i(U)$;
        \item For each $i,j\in\set{1,...,m}$ with $i\neq j$, $\phi_i(U)\cap \phi_j(U)=\emptyset$. 
    \end{enumerate}
\end{definition}
In particular, we will focus on relative fractal drums $(X,\Omega)$ such that $X$ is the attractor of a self-similar system $\Phi$ for which $\Omega$ is such an open set. 

When $X$ is a self-similar set, its tubular neighborhoods $X_\e$ may or may not themselves be self-similar in the same sense that it partitions according to a self-similar system. However, with an appropriate relative fractal drum an approximate self-similarity will still lend itself to direct analysis. Namely, if a relative fractal drum $(X,\Omega)$ (recall Definition~\ref{def:RFD}) can be written in terms of some similar RFDs $(\lambda_i X,\lambda_i\Omega)$ and a residual RFD $(X,R)$, then the tubular neighborhood can be partitioned according to its overlap with the sets $\lambda_i\Omega$ and $R$. 

To that end, suppose $\Phi$ is a self-similar system with attractor $X$, and let $\Omega$ be an open set for which $\Phi$ satisfies the open set condition and such that $(X,\Omega)$ is a relative fractal drum. Write $R=\Omega\setminus (\cup_{i=1}^m \phi_i(\Omega))$ so that $\Omega$ is the disjoint union of each $\phi_i(\Omega)$ and $R$. There are two points of order. Firstly, we need that $X_\e\cap \phi_i(\Omega)$ is in fact $(\phi_i(X))_\e\cap\phi_i(\Omega)$. Equivalently, for every point $y\in\phi_i(\Omega)$, $d(y,X)=d(y,\phi_i(X))$. This condition will be necessary to relate $V_{X,\phi_i(\Omega)}$ to $V_{\phi_i(X),\phi_i(\Omega)}$. 

The other is a small technical point: $R$ need not be open, so strictly speaking $(X,R)$ is not an RFD. If $\partial R$ has measure zero, we may safely replace $R$ with its interior and $V_{X,R}$ shall be unaffected. Even if not, however, the definition of $R$ ensures that it is measurable with $m(R)\leq m(\Omega)<\infty$, and thus $V_{X,R}$ (and consequently $\tubezeta_{X,R}$) will still be well-defined as in Definition~\ref{def:tubeFunction} (resp. Definition~\ref{def:tubeZeta}.)

We will conclude this subsection by endowing the relative fractal drums to which our methods apply with a suitable name. 
\begin{definition}[Osculating Sets of Iterated Function Systems]
    \label{def:oscRFD}
    Let $\Phi=\set{\phi_i}_{i=1}^m$ be an iterated function system on $\RR^\dimension$, and let $X$ be its attractor. An open set $\Omega\subset\RR^\dimension$ is said to be an osculating set for $\Phi$ if the following hold:
    \begin{enumerate}
        \item $\Omega\supset\bigcup_{i=1}^m \phi_i(\Omega)$;
        \item For each $i,j\in\set{1,...,m}$ with $i\neq j$, $\phi_i(\Omega)\cap \phi_j(\Omega)=\emptyset$;
        \item $\Omega$ has finite Lebesgue measure;
        \item $\exists\,\delta>0$ such that $\Omega_\delta \supset X$;
        \item For each $i=1,..,m$, if $y\in\phi_i(\Omega)$, then $d(y,X)=d(y,\phi_i(X))$. 
    \end{enumerate}
    The pair $(X,\Omega)$ is called an osculant fractal drum for $\Phi$. 
\end{definition}
Note that the first two conditions imply that $\Phi$ satisfies the open set condition as per Definition~\ref{def:OSC}. The third and fourth conditions imply that $(X,\Omega)$ is a relative fractal drum as per Definition~\ref{def:RFD}. The last condition, which to the best of our knowledge has not been studied before, might be called the ``osculating" condition: the points in iterates $\phi_i(\Omega)$ of $\Omega$ stay closest to the corresponding iterate $\phi_i(X)$ of $X$, rather than to a different iterate $\phi_j(X)$.

\subsection{Self-Similar Systems Yielding SFEs}
Given a self-similar system $\Phi$ which satisfies the open set condition, we suppose that it gives rise to an osculant fractal drum $(X,\Omega)$. For this work, our examples shall be focused on (self-avoidant) generalized von Koch snowflakes; see Section~\ref{sec:appToGKFs}. Under the conditions of Definition~\ref{def:oscRFD}, we deduce a scaling functional equation satisfied by the volume of a tubular neighborhood relative to $\Omega$. 
\begin{theorem}[SFE of a Self-Similar System]
    \label{thm:SFEofSSS}
    Let $\Phi=\set{\phi_i}_{i=1}^m$ be a self-similar system on $\RR^\dimension$ where the scaling ratio of $\phi_i$ is $\lambda_i$ for each $i=1,...,m$. Let $X$ be the attractor of $\Phi$ and suppose that $\Omega$ is an osculating set for $\Phi$. Define $R=\Omega\setminus(\cup_{i=1}^m \phi_i(\Omega)).$

    Then the tube function $V_{X,\Omega}:\RR^+\to[0,\infty]$ satisfies the following scaling functional equation:
    \begin{equation}
        \label{eqn:SFEofSSS}
        V_{X,\Omega}(\e) = L[V_{X,\Omega}](\e) + V_{X,R}(\e) = \sum_{i=1}^m \lambda_i^\dimension V_{X,\Omega}(\e/\lambda_i) + V_{X,R}(\e),
    \end{equation}
    for all $\e$. Here, $L:=\sum_{i=1}^m \lambda_i^\dimension M_{\lambda_i}$ is the associated scaling operator.
\end{theorem}
Note that a scaling ratio may be repeated, in which case $L$ may be written in the reduced form
\[ L = \sum_{j=1}^{m'} a_j \lambda_{j}^N M_{\lambda_j}, \]
where the $\lambda_j$'s are distinct, $a_j\in\NN$ is the multiplicity of $\lambda_j$, and $m'\leq m$. 
\begin{proof}    
    We begin by partitioning $\Omega$ according to the iterates and the residual set: $\Omega = \uplus_{i=1}^m \phi_i(\Omega) \uplus R$. Next, we intersect the sets of the partition with $X_\e$ and taking the Lebesgue measure. By disjoint additivity, we obtain that
    \[ V_{X,\Omega}(\e) = \sum_{i=1}^m V_{X,\phi_i(\Omega)}(\e) + V_{X,R}(\e). \]

    Now, we have that for each $i\in\set{1,...,m}$, $d(y,X)=d(y,\phi_i(X))$ for all $y\in\phi_i(\Omega)$. Consequently, $X_\e\cap\phi_i(\Omega)=(\phi_i(X))_\e\cap\phi_i(\Omega)$, as the epsilon neighborhoods are the preimages of the interval $[0,\e)$ under the respective distance functions $d(\cdot,X)\equiv d(\cdot,\phi_i(X))$, equivalent in $\phi_i(\Omega)$. Thus, $V_{X,\phi_i(\Omega)}=V_{\phi_i(X),\phi_i(\Omega)}$.
    
    Therefore, we obtain the identity that 
    \[ V_{X,\Omega}(\e) = \sum_{i=1}^m V_{\phi_i(X),\phi_i(\Omega)}(\e) + V_{X,R}(\e). \]
    Using Lemma~\ref{lem:volumeScaling}, we have that $V_{\phi_i(X),\phi_i(\Omega)}(t) = \lambda_i^\dimension V_{X,\Omega}(t/\lambda)$. Equation~\ref{eqn:SFEofSSS} then follows by applying this term-by-term. 
\end{proof}

In practice, it is generally easiest to ensure that a relative fractal drum is osculant in the construction of the set $\Omega$. In this work, we explicitly apply our results to generalized von Koch fractals in Section~\ref{sec:appToGKFs}. 

\subsection{Complex Dimensions from SFEs}
We first note the following which will serve as a dictionary between truncated Mellin transforms and tube zeta functions. 
\begin{proposition}[Tube Zeta Functions are Mellin Transforms]
    \label{prop:zetaMellin}
    Let $(X,\Omega)$ be a relative fractal drum in $\RR^\dimension$. Then 
    \[ \tubezeta_{X,\Omega}(s;\delta) = \Mm^\delta[t^{-\dimension} V_{X,\Omega}(t)](s). \]
\end{proposition}
The main note of interest is that the tube function is scaled by a factor of $t^{-\dimension}$. This in particular is behind the weights of $\lambda_i^\dimension$ in the scaling operator in Theorem~\ref{thm:SFEofSSS}. 

Now, given a scaling functional equation determined by a scaling operator $L$ and a remainder term for the tube function of a set $X$, we deduce the possible complex dimensions of $X$ up to an order based on the estimates for the remainder. 

\begin{theorem}[Complex Fractal Dimensions from SFEs]
    \label{thm:cdimsFromSFE}
    Let $(X,\Omega)$ be a relative fractal drum in $\RR^\dimension$ and let $V_{X,\Omega}(\e)$ be the corresponding relative tube function. Suppose that $V_{X,\Omega}$ satisfies the scaling functional equation
    $V_{X,\Omega}(\e) = L[V_{X,\Omega}](\e) + R(\e)$ where $L:=\sum_{i=1}^m a_i\lambda_i^\dimension M_{\lambda_i}$ is a scaling operator with distinct ratios $\lambda_i\in(0,1)$ with positive integral multiplicities $a_i$ and where $R(\e)$ is a continuous remainder term of order $O(\e^{\dimension-r})$, where $r\in\RR$ depends on this remainder quantity.
    
    Then the set of complex fractal dimensions with real part $\sigma>r$, i.e. in the open right half plane $\HH_r$, is contained in the set: 
    \begin{equation}
        \label{eqn:cdimsFromSFE}
        \Dd_X(\HH_r) \subset \Big\{\omega \in \HH_r : 1 = \sum_{i=1}^m a_i \lambda_i^\omega \Big\}.
    \end{equation}
    If $D=\overline{\dim}_\Mink(X,\Omega)$ is the relative upper Minkowski dimension of the set $X$, suppose that $r<D$. Then $\tubezeta_{X,\Omega}$ is holomorphic in $\HH_D$ with abscissa of absolute convergence $D(\tubezeta_{X,\Omega})=\overline{\dim}_\Mink(X,\Omega)$. 
\end{theorem}
Explicitly, the functional equation for $V_{X,\Omega}$ takes the form:
\begin{align}
    \label{eqn:tubeSFE}
    V_{X,\Omega}(\e) &= L[V_{X,\Omega}](\e) + R(\e) = \sum_{i=1}^{m} a_i\lambda_i^\dimension V_{X,\Omega}(\e/\lambda_i) + R(\e).
\end{align}
The structure of these complex (fractal) dimensions is dependent on the nature of the scaling ratios. In particular, there is a dichotomy between the lattice (or arithmetic) and the non-lattice (or non-arithmetic) cases. We note that the structure results are exactly the same as in Theorem~2.16 of \cite{LapvFr13_FGCD}. 
\begin{proof}
    First, we note that the function $f(t) = t^{-\dimension} V_{X,\Omega}(t)$ satisfies the corresponding scaling functional equation 
    \[ f(t) = \sum_{i=1}^m a_i f(t/\lambda_i) + t^{-\dimension} R(t), \]
    since $\lambda_i^{\dimension}\e^{-\dimension} V_{X,\Omega}(\e/\lambda_i)=f(\e/\lambda_i)$. Note that by Proposition~\ref{prop:zetaMellin}, the restricted Mellin transform of $f$ is exactly the tube zeta function of the RFD $(X,\Omega)$, i.e. $\Mm^\delta[f](s)=\tubezeta_{X,\Omega}(s;\delta)$. Let us also define $\tubezeta_R$ by $\tubezeta_R(s;\delta)=\Mm^\delta[t^{-\dimension}R(t)](s)$. 

    So, by Theorem~\ref{thm:transSFE}, we have that the tube zeta function $\tubezeta_{X,\Omega}$ satisfies the functional equation: 
    \begin{equation}
        \label{eqn:zetaFE}
        \tubezeta_{X,\Omega}(s;\delta) = \zeta_L(s)(E(s;\delta)+\tubezeta_R(s;\delta)).
    \end{equation}
    Here, $L=\sum_{i=1}^m a_i M_{\lambda_i}$, $\zeta_L$ is its associated scaling zeta function, and the remainder term satisfies $t^{-\dimension}R(t)=O(t^{-r})$ since $R(t)=O(t^{\dimension-r})$. This formula is valid for all $s\in \CC\setminus \Dd_L \cap \HH_r$. Additionally, the entire function $E$ can be written in terms of partial tube zeta functions, namely:
    \[ E(s) = E(s;\delta) = \sum_{i=1}^m a_i \lambda_i^s \partialzeta_{X,\Omega}(s;\delta,\delta/\lambda_i). \]
    The theorem additionally yields that $\tubezeta_{X,\Omega}(s;\delta)$ is holomorphic in $\HH_{\max(D,r)}$ and meromorphic in $\HH_r$, with $\Dd_X(\HH_r)\subset \Dd_L(\HH_L)$. Equation~\ref{eqn:cdimsFromSFE} follows from this containment and Proposition~\ref{prop:propertiesSZF}. 

    The abscissa of absolute convergence is a consequence of Theorem~4.1.7 in \cite{LRZ17_FZF}, in light of the functional equation relating tube and distance zeta functions. See Equation~2.2.23 and Theorem~5.3.2 in \cite{LRZ17_FZF}. 
\end{proof}

Note that Equation~\ref{eqn:zetaFE} yields the following additional information about the structure of $\Dd_X(\HH_r)$. Namely if $\Re(\omega)>r$, then $\omega\in\Dd_X(\HH_r)$ if and only if $\zeta_L$ has a pole of higher order than the degree of vanishing of the term $E(s;\delta)+\tubezeta_R(s;\delta)$. So, we obtain the following corollary. 
\begin{corollary}[Criterion for Exact Complex Dimensions]
    \label{cor:exactCdims}
    Let $(X,\Omega)$ be a relative fractal drum as in Theorem~\ref{thm:cdimsFromSFE}, in particular whose complex dimensions in $\HH_r$ satisfy Equation~\ref{eqn:cdimsFromSFE}. Suppose that for each $\omega\in\Dd_L(\HH_r)$ we have that 
    \[ \sum_{i=1}^m a_i\lambda_i^\omega \partialzeta_{X,\Omega}(\omega;\delta,\delta/\lambda_i) + \tubezeta_R(\omega;\delta) \neq 0. \]
    Then in fact $\Dd_X(\HH_r)=\Dd_L(\HH_r)$, where $\Dd_X(\HH_r)$ is the set of complex dimensions of the RFD $(X,\Omega)$ in $\HH_r$, as in Definition~\ref{def:cDims}. In other words, each such $\omega$ is a complex dimension of $(X,\Omega)$ rather than simply a possible complex dimension.
\end{corollary}
We conjecture that for every pole of $\zeta_L$, there is at least one value of $\delta$ for which the factor is nonzero. This claim would be enough to prove that these sets of poles are in fact equal, as the poles of $\tubezeta_{X,\Omega}$ are known to be independent of the parameter $\delta$. The integrals defining $\partialzeta_{X,\Omega}$ are notably oscillatory, so one must take care in estimating them and their linear combination. Additionally, this claim is likely dependent on a suitable estimate for $\tubezeta_R(s;\delta)$. 

Of note, Theorem~\ref{thm:cdimsFromSFE} is only as good as the scaling exponent in the known estimate of the remainder term. If the remainder term is merely bounded (viz. $R(t)=O(\e^0)$, with $r=\dimension$), then there is little information to be gained. The best results occur if the parameter $r$ is small; in Section~\ref{sec:appToGKFs}, our results shall apply for any order $\alpha>0$ arbitrarily close to $0$. Informally, $r$ corresponds to the (Minkowski) dimension of the remainder quantity: when $R$ scales like the volume of an epsilon neighborhood of a point, the exponent is $\dimension$ and thus $r=0$; when $R$ scales like the volume of a neighborhood of a line, so the exponent is $\dimension-1$ and $r=1$; and so forth. In other words, our work determines the complex dimensions with real parts (or amplitudes) up to the real part corresponding to the dimension of (the estimate of) the remainder term.

\section{Results on Generalized von~Koch Fractals}
\label{sec:appToGKFs}
%
%

In this section, we shall apply the results of Section~\ref{sec:fractalSFEs} to generalized von Koch fractals (or GKFs.) See Section~\ref{sec:GKFs} for the definitions and properties concerning these fractals. In particular, we provide a scaling functional equation for the tubular neighborhood volume and estimates of the error term that allow us to describe the possible complex dimensions of these GKFs. 

Notably, a lattice/non-lattice dichotomy arises depending on whether the scaling ratio $r$ and its conjugate factor $\ell=\frac12(1-r)$ are arithmetically related or not. For the standard von Koch snowflake, or the $(3,\frac13)$-von Koch snowflake to be precise, the ratios are the same: $r=\ell$, and thus the fractal falls under the lattice (aka arithmetic) case. The emergence of the non-lattice (aka non-arithmetic) case is unique to the generalizations of the von Koch snowflake. 

\subsection{Complex Dimensions of GKFs}
Let $\Knr$ be a $(n,r)$-von Koch snowflake which is a simple, closed curve.\footnote{See Proposition~\ref{prop:selfAvoid} for a sufficient condition based on $r$ and $n$.} By the Jordan curve theorem, a simple closed curve bisects the plane into two connected components, one of which is bounded and the other unbounded. Let us denote by $\Omega$ the bounded component; then $\partial\Omega=\Knr$. 

Note that $(\Knr,\Omega)$ is readily seen to be a relative fractal drum since the sets are both bounded. We shall be interested in describing the inner tubular neighborhood of $\Knr$, which is the tube function of $\Knr$ relative to the set $\Omega$. Our goal is to construct a tube functional equation \`a la Equation~\ref{eqn:tubeSFE} and apply Theorem~\ref{thm:cdimsFromSFE}. 

\begin{figure}[t]
    \centering
    \subfloat{\includegraphics[width=4cm]{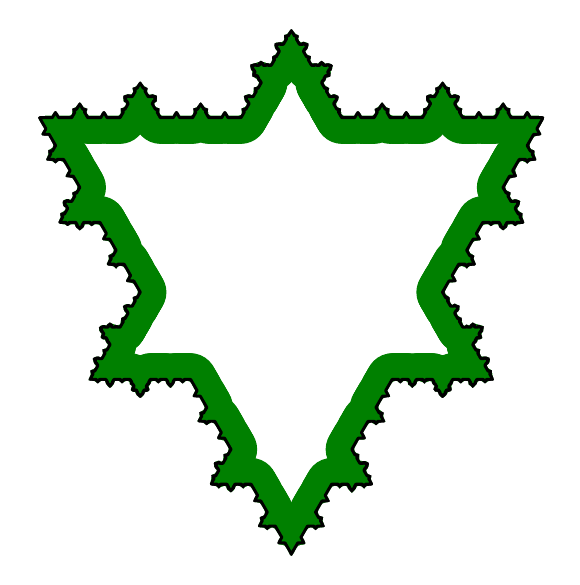}}
    \qquad
    \subfloat{\includegraphics[width=4cm]{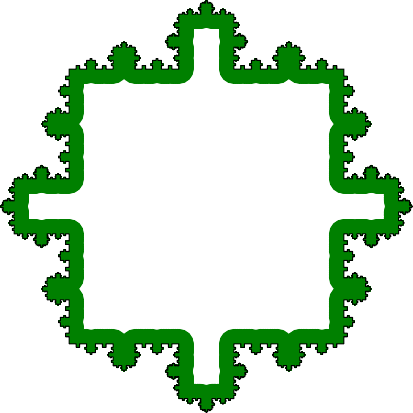}}
    \caption{Inner tubular neighborhoods of the fractals $K_{3,\frac15}$ (left) and $K_{4,\frac15}$ (right).}
    \label{fig:innerTubeNbds}
\end{figure}

To that end, we first note that by construction $\Knr$ has $n$-fold rotational symmetry about its center. So we may reduce our analysis to the portion of $(\Knr)_\e$ contained inside a single sector of angle $2\pi/n$ intersected with $\Omega$. So, let $S$ be any one of the $n$ congruent (open) sectors defined by the center of $\Knr$ and two adjacent vertices of the underlying regular $n$-gon. Defining $V_K$ for convenience, we have that 
\[ V_K(\e) := V_{\Knr,\Omega\cap S}(\e) = \frac1n V_{\Knr,\Omega}(\e)  \]
for all $\e\geq 0$. Let us also define a set $K(\e)$ which represents the portion of the inner epsilon neighborhood contained in $\Omega\cap S$, namely
\[ K(\e) := \{ y\in \Omega\cap S : \exists x\in K,\, |x-y|<\e  \}=(\Knr)_\e\cap \Omega\cap S. \]
We may partition $K(\e)$ into $2+(n-1)$ self-similar copies of itself as well as $4$ pieces contained within congruent triangles and $2$ pieces which are exactly circular sectors. This partitioning leads to the following theorem.
\begin{theorem}[Scaling Functional Equation of $V_K$]
    \label{thm:vonKochSFE}
    For any $\e\geq 0$, $V_K$ satisfies the approximate scaling functional equation
    \begin{equation}
        \label{eqn:volumeFunctionalEq}
        V_K(\e) = 2\ell^2V_K(\e/\ell)+(n-1)r^2V_K(\e/r)+R(\e),
    \end{equation}
    where $\ell=\frac12(1-r)$ is the conjugate scaling factor to $r$ and $R(\e)=O(\e^2)$ as $\e\to 0^+$.
    
    The remainder satisfies the estimate $R(\e)\leq (2\cot(\theta_n/2)+\theta_n)\e^2$, where $\theta_n=\frac{2\pi}n$ is the central angle of the regular $n$-gon, for all $\e>0$.   
\end{theorem}

\begin{proof}
    Without loss of generality, we may suppose that $K:=\Knr\cap S$ is exactly the $(n,r)$-von Koch curve $\Cnr$; in general, $K$ is isometric to this set. In this setting, we have explicitly the similitudes for the self-similar system $\Phi_{n,r}=\set{\phi_L,\phi_R,\psi_k, k=1,...,n-1}$ which defines $\Cnr$ as in Definition~\ref{def:vkCurve}. (In the general setting, simply compose these maps $\phi_i$ with the isometry between the sets.) 

    By definition, $K$ is the attractor of the self-similar system $\Phi$ in Definition~\ref{def:vkCurve}. There are two similitudes $\phi_L$ and $\phi_R$ with scaling ratios $\ell:=\frac12(1-r)\in(0,1)$ and $n-1$ similitudes $\psi_k$ with scaling ratios $r\in(0,1)$. So, we define the scaling operator $L=L_{n,r}$ by $L:= 2M_\ell + (n-1)M_r.$ 

    Let $U=\Omega\cap S$. Then we have that $U$ is an osculating set for $\Phi$. Indeed, the collection of sets $\set{\phi_L(U),\phi_R(U),\psi_k(U),k=1,...,n-1}$ is readily seen to be pairwise disjoint and each set is contained within $U$ itself. $(K,U)$ is a clearly a relative fractal drum since $K$ and $U$ are bounded sets. 
    
    The osculating condition may be seen from the symmetry involved in the construction of $K$. In particular, we shall use the property that if $X$ and $X'$ are two sets in $\RR^2$ related by reflection about the line $L$, then the set of points which are equidistant to $X$ and $X'$ must lie on $L$. We note that by choice of $U$, its images $\psi_k(U)$ partition a regular $n$-gon scaled by a factor of $r$ about the lines of symmetry passing from the center of the shape to its vertices. Because the adjacent images $\psi_j(K)$ and $\psi_{j+1}(K)$ are reflections of each other about these lines, it follows that $d(\cdot,K)=\min_{\phi\in\Phi}d(\cdot,\phi(K))$ can only change form across these lines by continuity of the distance functions. In other words, it follows that $d(\cdot,K)\equiv d(\cdot,\psi_j(K))$ in each $\psi_j(U)$, $j=1,...,n-1$. 

    Now, the other two iterates on the edges are simpler. Indeed, we have that $d(\cdot,K)\equiv d(\cdot,\phi_L(K))$ in $\phi_L(U)$ (resp. $d(\cdot,K)\equiv d(\cdot,\phi_R(K))$ in $\phi_R(U)$) is clear from the geometry of $K$ since a ball $B_\e(y)$ for $y\in\phi_L(U)$ (resp. $y\in\phi_R(U)$) will intersect with $\phi_L(K)$ (resp. $\phi_R(K))$ before ever reaching $\psi_1(K)$ (resp. $\psi_{n-1}(K)$). 
    
    So, by Theorem~\ref{thm:SFEofSSS}, we obtain the following equality:
    \[ V_K(\e) = 2\ell^2 V_K(\e/\ell) + (n-1)r^2V_K(\e/r) + V_{K,R}(\e), \]
    where $R=U\setminus(\cup_{\phi\in\Phi} \phi(U))$ and $V_{K,R}(\e) = m(K(\e)\cap R)$. Now, we note that $K(\e)\cap R$ may be partitioned into two sectors of angle $\theta_n$ and length $\e$ and four congruent pieces contained in triangles of height $\e$ and width $\e\cot(\theta_n/2)$. See Figure~\ref{fig:trianglet} for a depiction of one such triangle. Therefore, we may write that
    \[ R(\e) := V_{K,R}(\e) \leq \theta_n\e^2+2\e^2\cot(\theta_n/2). \]
\end{proof}

\begin{figure}[t]
    \centering
    \includegraphics[width=0.6\textwidth,trim=10 20 10 0, clip]{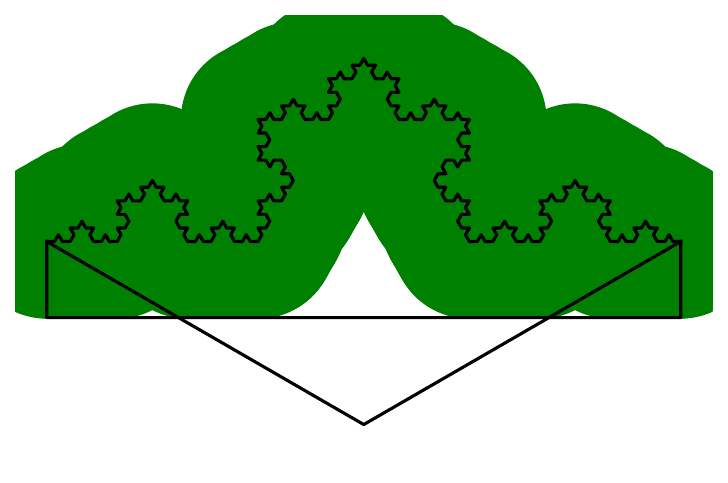}
    \caption{A portion of the epsilon neighborhood of the von Koch snowflake $K_{3,1/3}$ is depicted. In particular, note the regions contained in the triangles of height $\e$ and width $\e\cot(\theta_n/2)$, where $\theta_n=2\pi/n$, on the left and the right.}
    \label{fig:trianglet}
\end{figure} 

Note that it follows that $V_{\Knr,\Omega}$ satisfies the same functional equation as Equation~\ref{eqn:volumeFunctionalEq} but with remainder $nR(\e)$. So, Theorem~\ref{thm:cdimsFromSFE} will yields a description of the complex dimensions of $\Knr$. Additionally, we highlight the description of the tube zeta function in terms of a self-similar zeta function, $\zeta_L$.
\begin{theorem}[Possible Complex Dimensions of GKFs]
    \label{thm:cDimsOfGKFs}
    Let $n\geq3$ and $r\in(0,1)$ be such that $\Knr$ is a simple, closed curve. Let $\tilde\zeta_{\Knr,\Omega}(s)$ be the relative tube zeta function for $\Knr$ relative to the bounded component $\Omega$ of $\RR^2\setminus\Knr$. Further, let $\ell=\frac12(1-r)$ be the conjugate scaling ratio to $r$ and $\delta>0$. Then for all $s\in\CC\setminus \Dd_L\cap\HH_0$, 
    \begin{equation}
        \label{eqn:zetaGKF}
        \tilde\zeta_{\Knr,\Omega}(s) = \cfrac1{1-2\ell^s-(n-1)r^s}\cdot \Big( E(s;\delta) + \tubezeta_R(s;\delta) \Big),
    \end{equation}
    where $E(s;\delta)$ is an entire function, $\tubezeta_R(s;\delta):=n\int_0^\delta t^{s-3}R(t)\,dt$ is holomorphic in the right half plane $\HH_0$, and $L:=2M_\ell+(n-1)M_r$ is a scaling operator with associated zeta function $\zeta_L(s)=\frac1{1-2\ell^s-(n-1)r^s}$ and complex dimensions $\Dd_L(\CC)=\set{\omega\in\CC:1=2\ell^\omega+(n-1)r^\omega}$. 

    Consequently, we have that the complex dimensions of $\Knr$ satisfy the containment: 
    \[ \Dd_{\Knr}(\HH_0) \subseteq \{ \omega \in \CC : 1-2\ell^\omega-(n-1)r^\omega=0 \}. \]
    $\tubezeta_{\Knr,\Omega}$ is meromorphic in $\HH_0$ with poles at the points in $\Dd_{\Knr}(\HH_0)$ and holomorphic in $\HH_D$, where $D=\overline{\dim}_\Mink(\Knr,\Omega)$. 
\end{theorem}
\begin{proof}
    This is a corollary of Theorem~\ref{thm:vonKochSFE} and Theorem~\ref{thm:cdimsFromSFE}. See in particular Equation~\ref{eqn:zetaFE} from the proof of the latter. Explicitly, we have that 
    \[ E(s;\delta) = 2\ell^s\partialzeta_{\Knr,\Omega}(s;\delta,\delta/\ell)+(n-1)r^s\partialzeta_{\Knr,\Omega}(s;\delta,\delta/r) \]
    and that $R(\e)\leq (\theta_n+2\cot(\theta_n/2))\e^2$, where $\theta_n=2\pi/n$. 
\end{proof}

\subsection{Languidity of Generalized von Koch Fractals}
In this section, we provide estimates on the growth of $\tubezeta_{\Knr,\Omega}$ known as languidity, see Definitions~5.1.3 and 5.1.4 in \cite{LRZ17_FZF}. These are primarily necessary in order to deduce tube formulae for $V_{\Knr,\Omega}$. Note that here we use $i$ for the imaginary unit, whereas in previous sections we have used it as an indexing variable. 

\begin{theorem}[Languidity of GKF Relative Fractal Drums]
    \label{thm:languidZetaGKF}
    Let $(\Knr,\Omega)$ be the relative fractal drum of a GKF as in Theorem~\ref{thm:cDimsOfGKFs}, with tube zeta function $\tubezeta_{\Knr,\Omega}$. Then $\tubezeta_{\Knr,\Omega}$ is languid with exponent $\kappa=0$. If, a fortiori, $\tubezeta_R$ is strongly languid as a function (in the sense of Definition~\ref{def:stronglyLanguid}) then so too is $\tubezeta_{\Knr,\Omega}$ and the RFD $(\Knr,\Omega)$.
\end{theorem}
\begin{proof}
    Firstly, we note that by Proposition~\ref{prop:languidSZF}, $\zeta_L$ is strongly languid with exponent $\kappa=0$. We shall thus need to estimate the terms in the multiplicative factor in order to deduce the languidity of $\tubezeta_{\Knr,\Omega}$. We do have that 
    \begin{align*}
        |\tubezeta_{\Knr,\Omega}(s;\delta)| &\leq |\zeta_L(s)|(|E(s;\delta)|+|\tubezeta_{R}(s;\delta)|),
    \end{align*}
    so we may estimate $E(s;\delta)$ and $\tubezeta_R(s;\delta)$ independently. 

    Let $\alpha\in(0,D)$, $\beta>2\geq D$, and $\delta>0$ be fixed. Let us define the screen $S(\tau)$ to be the constant function $S(\tau)\equiv \alpha$. Further, suppose that $\set{\tau_m}_{m\in\ZZ}$ is a sequence of real numbers for which $\zeta_L$ satisfies the languidity estimates. In particular, it satisfies that $\lim_{m\to\infty}\tau_m\to\infty$, $\lim_{m\to-\infty}\tau_m=-\infty$, and $\forall m\geq 1$, $\tau_{-m}<0<\tau_m$. 
     
    To estimate the entire function $E(s;\delta)$, we recall the estimates in Equation~\ref{eqn:partialTubeEstimate} for partial tube zeta functions. Using these estimates on the function $E(s;\delta)$, when $\sigma=\Re(s)\neq 2$ we have that:
    \begin{align*}
        |E(s;\delta)| &\leq 2 \ell^\sigma V_{\Knr,\Omega}(\delta/\ell) \frac{\delta^{\sigma-2}}{\sigma-2}(\ell^{2-\sigma}-1) \\
            & + (n-1) r^\sigma V_{\Knr,\Omega}(\delta/r) \frac{\delta^{\sigma-2}}{\sigma-2}(r^{2-\sigma}-1) \\
            &\leq C \frac{\delta^{\sigma-2}}{\sigma-2}(\ell^{2}-\ell^\sigma+r^{2}-r^\sigma).
    \end{align*} 
    Additionally, $E(s;\delta)$ is bounded in a neighborhood of $\sigma = 2$, say $[2-\gamma,2+\gamma]$, by the second estimate of Equation~\ref{eqn:partialTubeEstimate} and the continuity of $E(\cdot;\delta)$. So, $E(s;\delta)$ may be uniformly estimated on the intervals $[\alpha+i\tau_m,(2-\gamma)+i\tau_m]$, $[2-\gamma+i\tau_m,2+\gamma+i\tau_m]$, and $[2+\gamma+i\tau_m,\beta+i\tau_m]$ with a constant $C$ which does not depend on $\tau_m$. Along the screen $S(\tau)\equiv\alpha$, we similarly have that $E(s;\delta)$ is uniformly estimated by a constant determined only by $\alpha$ and the other fixed parameters ($\delta,n,r$ and so forth.) We do note that the estimate involves a factor of the form $B^\alpha$ for some $B>0$, but otherwise is valid for any $\alpha<D$; thus $E(s;\delta)$ may in fact be seen to satisfy the estimates required of strongly languidity.  

    Next, we estimate the remainder term. By Theorem~\ref{thm:vonKochSFE}, we have the estimate $nR(t)\leq C t^2$ for a constant $C>$ depending only on $n$. Therefore, we have that if $\Re(s)=\sigma>0$,
    \[ |\tubezeta_R(s;\delta)| \leq C \int_0^\delta t^{\sigma-1}\,dt \leq C \frac{\delta^\sigma}\sigma . \]  
    Therefore, we have that it is bounded on $\Re(s)=\alpha>0$ independently of $\Im(s)$ (with bound depending only on $n$ and $\alpha$) and on the intervals $[\alpha+i\tau_m,\beta+i\tau_m]$ independently of $\tau_m$. 
    
    Let $C_L$ be a constant for which the estimates on $\zeta_L$ hold, $C_E$ for $E$, and $C_R$ for $\tubezeta_R$. Combining these estimates together with the constant $C=C_L(C_E+C_R)$ yields the estimates for the function $\tubezeta_{\Knr,\Omega}$. Thus, the RFD $(\Knr,\Omega)$ is languid with exponent $\kappa=0$ with respect to any screen $S(\tau)\equiv \alpha>0$. Additionally, if $\tubezeta_R$ satisfies the strong languidity conditions, then so too does $\tubezeta_{\Knr,\Omega}$ since the functions $\zeta_L$ and $E$ both satisfy the necessary estimates.  
\end{proof}

\subsection{Tube Formulae for GKFs}
Having proven the necessary growth estimates, we may now deduce explicit tube formulae for the RFDs $(\Knr,\Omega)$. We shall focus on pointwise formulae in this paper for simplicity, but the interested reader may also wish to consider the distributional analogues which are simpler in expression but more technical in their interpretation. There is an explicit formula for the tube function $V_{\Knr,\Omega}$ which is valid distributionally (Equation~\ref{eqn:distributionGKFtubeFormula}), and the pointwise formulae of Theorem~\ref{thm:GKFtubeFormula} will be valid for the integrated tube functions $V_{\Knr,\Omega}^{[k]}$ for any $k\geq 2$.

In what follows, we shall use the Pochhammer symbol $(s)_k$ which may be defined as
\[ (s)_k = \frac{\Gamma(s+k)}{\Gamma(s)}, \]
which for $k\in\NN$ takes the form $(s)_k = s(s+1)\cdots(s+k-1)$. 

\begin{theorem}[Pointwise Tube Formula for GKFs]
    \label{thm:GKFtubeFormula}
    Let $(\Knr,\Omega)$ be the relative fractal drum of a generalized von Koch snowflake relative to its interior component $\Omega$. Let $V_{\Knr,\Omega}$ be its tube function and let $V_{\Knr,\Omega}^{[k]}$ be the $k\nth$ antiderivative of $V_{\Knr,\Omega}$ with the constraint that $V_{\Knr,\Omega}^{[l]}(0)=0$ for all $l=1,...,k$. 

    Let $\delta>0$ and $k>1$. Then for any $t\in(0,\delta)$, we have that 
    \begin{align}
        \label{eqn:GKFtubeFormula}
        V_{\Knr,\Omega}^{[k]}(t) 
            &= \sum_{\mathclap{\omega\in\Dd_{\Knr,\Omega}(\HH_\alpha)}} 
            \,\Res\Bigg( \frac{t^{2-s+k}}{(3-s)_k} \tubezeta_{\Knr,\Omega}(s;\delta);\omega \Bigg) 
            + O(t^{2-\alpha+k}),
    \end{align}
    as $t\to0^+$ for any $\alpha>0$ for which $\set{\omega\in \Dd_{\Knr,\Omega}(\HH_0) : \Re(\omega) =\alpha}=\emptyset$. 
\end{theorem}

\begin{proof}
    By Theorem~\ref{thm:languidZetaGKF}, the RFD $(\Knr,\Omega)$ is languid. Therefore, we may apply Theorem~5.1.11 in \cite{LRZ17_FZF} to obtain Equation~\ref{eqn:GKFtubeFormula}. Explicitly, we construct the screen $S(\tau):=\alpha + i\tau$, which by assumption does not intersect any poles of $\tubezeta_{\Knr}$. In this case, $\sup S=\alpha$ and the error estimate follows from Equation~5.1.33. 
\end{proof}
Note that if a pole $\omega\in \Dd_{\Knr,\Omega}(\HH_0)$ is simple, then we have that 
\[ 
    \Res\Bigg( \frac{t^{2-s+k}}{(3-s)_k} \tubezeta_{\Knr,\Omega}(s;\delta);\omega \Bigg) 
        = \frac{t^{2-\omega+k}}{(3-\omega)_k} \Res(\zeta_L) (E(\omega;\delta)+\tubezeta_R(\omega;\delta)), 
\] 
where $\zeta_L$, $E$, and $\tubezeta_R$ are as in Theorem~\ref{thm:cDimsOfGKFs}. In this case, the formula is seen to be a sum of terms of the form $\frac{r_\omega}{(3-\omega)_k}t^{2-\omega+k}$, where $r_\omega$ is a constant determined from $\tubezeta_{\Knr}$. Note that this computation of the residue relies on Equation~\ref{eqn:zetaGKF} and the holomorphicity of $E$ and $\tubezeta_R$. 

Also, we note that the Pochhammer symbol has a shift by one since the term in Equation~5.1.30 in \cite{LRZ17_FZF} is $(\dimension-s+1)_k$; the codimension $2-\omega$ is still the important term in the denominator. Indeed, by writing the product formula in reverse order and fixing the codimension $2-\omega$ in each factor, we remark that it takes the form 
\[ (3-\omega)_k = (2-\omega+k)(2-\omega+k-1)\cdots(2-\omega+1). \]
In this form, the correspondence between the exponent of $t$ and the factors in the product is clearer.

Lastly, we stress that the sum in Equation~\ref{eqn:GKFtubeFormula} must be interpreted as an appropriate limit of sums over poles contained in windows of increasing size. More explicitly, it is the limit of the sum of the complex dimensions whose imaginary parts are bounded in magnitude by a fixed upper bound, say $T_n$, with the limit taken over the bounding parameter $T_n$ going to infinity. Equivalently, this may be stated in terms of a limit of sums over truncated windows. See Remark~5.1.12 in \cite{LRZ17_FZF} for more details. 

Now, Equation~\ref{eqn:GKFtubeFormula} is explicitly for second order and higher antiderivatives of the tube function $V_{\Knr,\Omega}$. This formula is, however, valid in a distributional sense for all $k\in \ZZ$ (and in particular for the tube function itself where $k=0$) without further assumption; see Theorem~5.2.2 in \cite{LRZ17_FZF}. Interpreting the following equality in the distributional sense, the formula for $k=0$ takes the form:
\begin{equation} 
    \label{eqn:distributionGKFtubeFormula}
    V_{\Knr,\Omega}(t) = \sum_{\mathclap{\omega\in\Dd_{\Knr,\Omega}(\HH_\alpha)}} 
    \,\Res(t^{2-s} \tubezeta_{\Knr,\Omega}(s;\delta);\omega ) + \Rr_{\Knr,\Omega}(t), 
\end{equation}
where $\Rr_{\Knr,\Omega}$ is a distribution whose action on test functions is of asymptotic order $O(t^{2-\alpha})$ for small $\alpha>0$, in the sense of Definition~5.2.9 in \cite{LRZ17_FZF}. If the poles $\omega$ are each known to be simple, then Equation~\ref{eqn:distributionGKFtubeFormula} takes the simpler form 
\begin{equation} 
    \label{eqn:distributionGKFtubeFormulaSimple}
    V_{\Knr,\Omega}(t) = \sum_{\mathclap{\omega\in\Dd_{\Knr,\Omega}(\HH_\alpha)}} 
    \,a_\omega t^{2-\omega}  + \Rr_{\Knr,\Omega}(t), 
\end{equation}
where $a_\omega := \Res(\tubezeta_{\Knr,\Omega}(s;\delta);\omega )$ are constants determined only by the residues of the tube zeta function. We refer the interested reader to Section~5.2 of \cite{LRZ17_FZF} for more detail and most importantly the description of the appropriate space of test functions for which these distributions act. Theorem~5.2.2 in particular leads Equation~\ref{eqn:distributionGKFtubeFormula} and Equation~\ref{eqn:distributionGKFtubeFormulaSimple} follows from Corollary~5.2.12.

\subsection{Lattice/Non-Lattice Dichotomy}
The possible complex dimensions $\Dd_L(\HH_0)$, or equivalently the complex solutions to Equation~\ref{eqn:Moran} with positive real part, have structure heavily dependent on the arithmetic properties of $r$ and its conjugate $\ell$. We shall say that $\Knr$ is arithmetic (or lattice) if $\log r/\log \ell$ is rational and non-arithmetic (or non-lattice) otherwise. In the arithmetic case, this implies there exists a (unique) positive number $x$ and positive integers $p,q$ so that $r=x^p$ and $\ell=x^q$.

\begin{corollary}[Complex Dimensions of $\Knr$, Arithmetic Case]
    Let the scaling ratios $r$ and $\ell=\frac12(1-r)$ be arithmetic, that is, there exist a positive $x$ and positive integers $p,q$ so that $r=x^p$, $\ell=x^q$. 

    Then there exist finitely many poles $D=\omega_1, \omega_2,...,\omega_n$ with real parts $D=\sigma_1>\sigma_2\geq ... \geq \sigma_n$ so that the complex dimensions of $\Knr$ are contained in the set 
    \[ \Big\{ \omega_k + \frac{2\pi i n}{\log x^{-1}} : n\in\ZZ, k=1,...,n \Big\}. \]
    For each $\omega$ with $\Re(\omega)_k>0$, these are indeed complex dimensions, and all of the poles with real part $D$ are simple. 
\end{corollary}
\begin{proof}
    This result is a corollary of Theorem~\ref{thm:cDimsOfGKFs} and the results in \cite[Theorem 2.16]{LapvFr13_FGCD} regarding self-similar zeta functions having the same form as $\zeta_L$. 
\end{proof}
This corollary implies that whenever $\Knr$ is arithmetic, it is not Minkowski measurable. 

\begin{figure}[ht!]
    \centering
    \begin{tabular}{ccc}
        $K_{3,\frac13}$ & $K_{4,\frac14}$ & $K_{5,\frac15}$ 
        \vspace{0.1cm}\\
        \includegraphics[width=3cm]{figures/K_n=3,r=0.33,lv=6.pdf}
        &
        \includegraphics[width=3cm]{figures/K_n=4,r=0.25,lv=6.pdf}
        &
        \includegraphics[width=3cm]{figures/K_n=5,r=0.20,lv=6.pdf}
        \\
        \includegraphics[width=3cm]{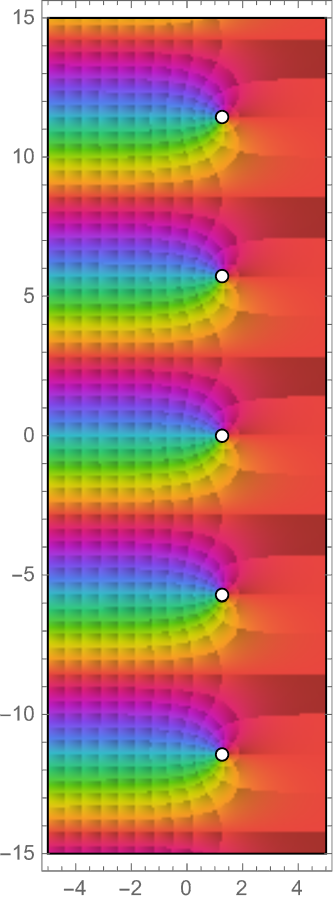}
        &
        \includegraphics[width=3cm]{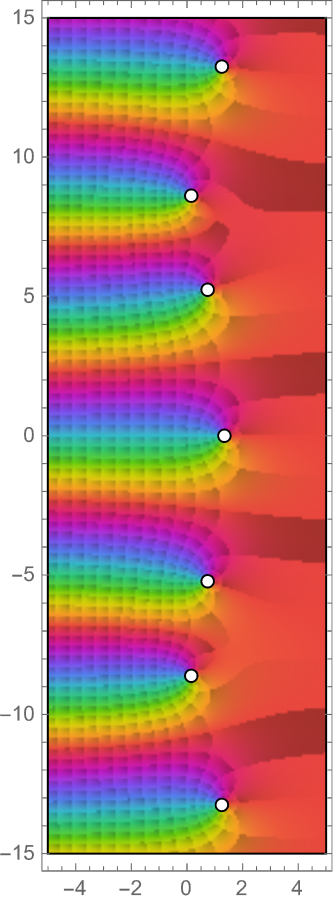}
        &
        \includegraphics[width=3cm]{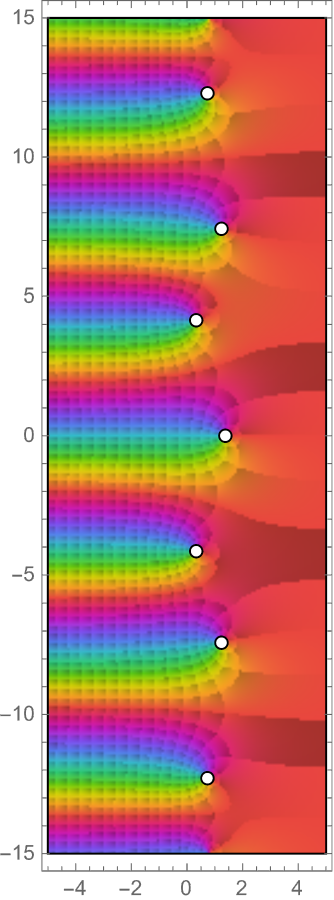} 
    \end{tabular}
    \caption{Two dimensional complex argument plots of associated scaling zeta functions of the fractals $K_{3,\frac13}$ (left), $K_{4,\frac14}$ (middle), and $K_{5,\frac15}$ (right). The possible complex dimensions of the fractals occur at the poles of these functions, indicated in the plots by the white circles.}
    \label{fig:cDimPlots2D}
\end{figure}

Meanwhile, in the non-lattice case, the complex dimensions are no longer distributed periodically along finitely many vertical lines. Instead, there is a unique pole of highest order, the Minkowski dimension $D$ of the fractal or equivalently the unique positive real solution to Equation~\ref{eqn:Moran} specialized to these fractals. Compare the possible complex dimensions of the fractals plotted in Figure~\ref{fig:cDimPlots2D}: the first example is the lattice case since there is one scaling ratio $r=\frac13$. However, in the other cases, the ratios of logarithms of scaling ratios take the form $\log(1/4)/\log(3/8)$ for $K_{4,\frac14}$ and $\log(1/5)/\log(2/5)$ for $K_{5,\frac15}$. 

Describing these poles exactly is in general quite challenging since they are solutions to a transcendental equation without the simplifications possible in the arithmetic case. However, the possible complex dimensions in the non-lattice case have a quasiperiodic structure which in particular may be approximated by sets of lattice complex dimensions. We refer the reader to Chapter~3, and in particular Theorem~3.6 of \cite{LapvFr13_FGCD} for a thorough description of the structure of non-lattice complex dimensions and the proofs of these statements. The theory of approximation of non-lattice complex dimensions may be found in Section~3.4 of \cite{LapvFr13_FGCD} and also in the further study of Lapidus, van Frankenhuijsen, and Voskanian in \cite{LvFV21}.

\section{Conclusion}
\label{sec:conclusion}
%
%

In summary, we have analyzed fractals arising from self-similar systems by constructing scaling functional equations. These scaling functional equations may be solved by means of truncated Mellin transforms. In the case of tube formulae for relative fractal drums, this process induces a functional equation satisfied by the zeta function. Moreover, the poles of this zeta function are essentially controlled by the self-similar complex dimensions from the scaling functional equation. 

We conclude the work by discussing future directions and open problems related to this work.

\subsection{Future Directions}
Firstly, the methods here may be applied to compute the possible complex dimensions of many other self-similar fractals. All of the necessary framework is present in Sections~\ref{sec:generalSFEs} and \ref{sec:fractalSFEs}, and one primarily needs to verify that an RFD $(X,\Omega)$ is osculant and languid in order to deduce an explicit tube formula such as Equation~\ref{eqn:GKFtubeFormula} for generalized von Koch fractals. 

Secondly, it is of interest to more carefully analyze the exactness of complex dimensions obtained via this methodology. Under what conditions are the possible complex dimensions exact, such as when does Corollary~\ref{cor:exactCdims} apply? We conjecture that a much stronger conclusion is possible, namely that Equation~\ref{eqn:cdimsFromSFE} is in fact an equality with only mild assumptions on the remainder term. 

Additionally, the analysis of generalized von Koch fractals has been done in this specific way in order to better understand the connection between geometry and spectrum of fractals in higher dimensions. For one-dimensional fractal strings, the relationship is elegantly given by an identity of geometric and spectral zeta functions using Riemann's zeta function (see \cite{LapvFr13_FGCD}.) The relationship is being explored in higher dimensions, and while there are known results (see for example Chapter~4, Section~3 in \cite{LRZ17_FZF}, and explicitly Theorems~4.3.8 and 4.3.11), the general problem is open. See Problem~6.2.32 in \cite{LRZ17_FZF} (and the references therein) for a more thorough discussion of this problem. 

In particular, we have studied geometric properties of generalized von Koch fractals by means of scaling functional equations in order to compare with the work of Michiel van den Berg such as in \cite{vdB00,vdBGil98,vdBHol99} on the heat equation on such fractals. We plan to further explore this connection in future work.

\bigskip
\noindent \textbf{Acknowledgements.} We wish to thank Dr. Michel Lapidus for his mentorship and helpful conversations regarding this work.


\bibliographystyle{amsalpha}
\bibliography{hoffer}

\noindent
\textbf{Will Hoffer}

\noindent
Department of Mathematics, University of California, Riverside, 900 University Avenue, Riverside, CA 92521-0135, USA. 

\noindent
\textit{Email address:} \href{mailto:whoff003@ucr.edu}{whoff003@ucr.edu}

\noindent
\textit{Website:} \url{https://willhoffer.com}

\end{document}